%% file: main.tex
\documentclass[runningheads]{llncs}

\newif\ifpublish
\publishtrue

\input{macros.tex}

\begin{document}

\title{\sysname: Low-Latency Asynchronous BFT DAG-Based Consensus}

\ifpublish
  \author{
    Philipp Jovanovic\inst{1} \and
    Lefteris Kokoris-Kogias\inst{2} \and
    Bryan Kumara\inst{3} \and
    Alberto Sonnino\inst{1,2} \and
    Pasindu Tennage\inst{4},
    Igor Zablotchi\inst{2}
  }
  \authorrunning{P. Jovanovic et al.}

  \institute{
    University College London (UCL) \and
    Mysten Labs \and
    Turing Institute \and
    EPFL
  }
\else
  \author{}
  \institute{}
\fi

\maketitle

\begin{abstract}
  \input{sections/abstract.tex}
\end{abstract}

\input{sections/introduction.tex}
\input{sections/overview.tex}
\input{sections/protocol.tex}
\input{sections/implementation.tex}

\input{sections/evaluation.tex}
\input{sections/related-work.tex}
\input{sections/conclusion.tex}

\ifpublish
  \subsubsection*{\ackname}
  \input{acks.tex}
\fi

\bibliographystyle{splncs04}
\bibliography{references}

\appendix

\input{sections/algorithms.tex}
\input{sections/example.tex}
\input{sections/proofs.tex}
\input{sections/appendix-eval.tex}

\end{document}

%% file: macros.tex
\usepackage[T1]{fontenc}
\usepackage{graphicx}
\usepackage{amsfonts}
\usepackage{amsmath}
\usepackage{amssymb}
\usepackage{bm}
\usepackage{todonotes}
\usepackage{xcolor}
\usepackage{xspace}
\usepackage{url}
\usepackage{algorithm}
\usepackage[noend]{algpseudocode}
\usepackage[hidelinks]{hyperref}
\usepackage{cleveref}
\usepackage{enumitem}
\setlist[itemize]{leftmargin=*}
\usepackage[font=footnotesize]{caption}
\usepackage{libertine}
\usepackage{subcaption}
\usepackage[scaled=0.85]{beramono} 

\ifdefined\publish
    \newcommand{\missing}[1]{}
    \newcommand{\alberto}[1]{}
    \newcommand{\bryan}[1]{}
    \newcommand{\lefteris}[1]{}
    \newcommand{\pasindu}[1]{}
    \newcommand{\philipp}[1]{}
\else
    \newcommand{\missing}[1]{\textcolor{red}{#1}}
    \newcommand{\alberto}[1]{\textcolor{orange}{\textbf{Alberto:} #1}}
    \newcommand{\bryan}[1]{\textcolor{blue}{\textbf{Bryan:} #1}}
    \newcommand{\lefteris}[1]{\textcolor{yellow}{\textbf{Lefteris:} #1}}
    \newcommand{\pasindu}[1]{\textcolor{green}{\textbf{Pasindu:} #1}}
    \newcommand{\philipp}[1]{\textcolor{purple}{\textbf{Philipp:} #1}}
\fi

\newcommand{\codelink}{
    \ifpublish
        \url{https://github.com/PasinduTennage/mahi-mahi-consensus} (commit \texttt{13a2819800ad3b192ab3c54dfe3ec4d0b34ae361})
    \else
        Code available but link omitted for blind review.
    \fi
}

\newcommand{\pparagraph}[1]{\medskip \noindent \textbf{#1.}}
\newcommand{\para}[1]{\pparagraph{#1}}

\newcommand{\etal}{\textit{et al.} }

\newcommand{\ie}{\textit{i.e.,} }

\newtheorem{observation}{Observation}

\newcommand{\algsize}{\scriptsize}

\newcommand{\sysname}{\textsc{Mahi-Mahi}\xspace}
\newcommand{\mysticeti}{\textsf{Mysticeti}\xspace}

\newcommand{\wavelength}{\texttt{waveLength}\xspace}
\newcommand{\numleaders}{\texttt{leadersPerRound}\xspace}
\newcommand{\decider}{\ensuremath{\mathsf{D}}\xspace}
\newcommand{\status}{\ensuremath{status}\xspace}
\newcommand{\waveoffset}{\texttt{waveOffset}\xspace}
\newcommand{\leaderoffset}{\texttt{leaderOffset}\xspace}
\newcommand{\self}{\ensuremath{\mathsf{Self}}\xspace}
\newcommand{\store}{\ensuremath{DAG}\xspace}

\newcommand{\bcast}[2]{\ensuremath{\textsf{a\_bcast}_#1}(#2)}
\newcommand{\deliver}[2]{\ensuremath{\textsf{a\_deliver}_#1}(#2)}

\newcommand{\propose}{\textsf{Propose}\xspace}
\newcommand{\boost}{\textsf{Boost}\xspace}
\newcommand{\vote}{\textsf{Vote}\xspace}
\newcommand{\certify}{\textsf{Certify}\xspace}

\newcommand{\scommit}{\texttt{commit}\xspace}
\newcommand{\sskip}{\texttt{skip}\xspace}
\newcommand{\sundecided}{\texttt{undecided}\xspace}

%% file: sections/abstract.tex
We present \sysname, the first asynchronous BFT consensus protocol that achieves sub-second latency in a WAN setting while processing over 100,000 transactions per second.
\sysname achieves such high performance by leveraging an uncertified structured Directed Acyclic Graph (DAG) to forgo explicit certification.
This reduces the number of messages required to commit and the CPU overhead for certificate verification significantly.
\sysname introduces a novel commit rule that enables committing multiple blocks in each asynchronous DAG round.
\sysname can be para\-metrized either with a 5 message commit delay,
maximizing the commit probability under a continuously active asynchronous adversary, or with a 4 message commit delay, reducing latency under a more moderate and realistic asynchronous adversary.
We demonstrate safety and liveness of \sysname in a Byzantine context for all of these parametrizations.
Finally, we evaluate \sysname in a geo-replicated setting and compare its performance to state-of-the-art asynchronous consensus protocols, showcasing \sysname's significantly lower latency.


%% file: sections/introduction.tex
\section{Introduction}\label{sec:introduction}


Consensus enables a set of replicas to agree on a common value from an initial
set of proposals, even in the presence of failures or malicious actors~\cite{cachin2011introduction}.
It is a fundamental primitive for maintaining strong data consistency between
replicas in distributed or decentralized systems.


Applications that require Byzantine Fault Tolerant (BFT)
consensus~\cite{sok-consensus}, such as blockchains~\cite{chainspace,sok-consensus,fastpay,sui-lutris,kokoris17omniledger,sonnino2020replay},
often rely on protocols designed for the partially syn\-chro\-nous network model
which aims to approximate mostly benign network conditions and allows the system to
perform well under these circumstances.
However, protocols designed for partial synchrony lose liveness under
asynchronous conditions, which can arise from poor connectivity, an active network adversary, or denial-of-service (DoS) attacks~\cite{consensus-dos}.
Asynchronous consensus protocols~\cite{cachin2011introduction,thunderbolt,ren2017}
address this issue by providing as much liveness as the network connectivity
allows.
To achieve this, these protocols sacrifice performance during periods
of network synchrony, resulting in significantly higher latency compared to
their partially synchronous counterparts. While state-of-the-art partially
synchronous protocols can process over $100,000$ transactions per second with
sub-second WAN latency~\cite{shoal++,mysticeti}, current asynchronous protocols achieve similar throughput with latencies on the order of seconds~\cite{narwhal-tusk}. 
This substantial latency drawback has made asynchronous consensus protocols less attractive for practical deployment.


Dual-mode protocols~\cite{jolteon,bullshark} attempt to provide the best of both worlds by operating partially-synchronous consensus by default and reverting to a less performant asynchronous sub-protocol when network conditions become adverse.
However, these dual-mode protocols introduce complexity and are prone to errors, as they must maintain two separate protocol stacks and implement mechanisms to detect changing network conditions and switch between the two consensus modes. Additionally, they remain vulnerable to targeted attacks that can cause the protocol to switch constantly between the two modes~\cite{bullshark}.
Due to these drawbacks, no dual-mode protocol has yet been deployed in a production environment to the best of our knowledge.


We therefore ask if it is possible to design a protocol that can simultaneously:
(i) provide liveness under asynchronous network
conditions,
(ii) achieve performance comparable to state-of-the-art partially-synchronous consensus protocols, and
(iii) maintain a simple design that allow for effective security analysis,
implementation, and maintenance?


In this paper, we introduce \sysname, a novel low-latency and high-throughput
asynchronous consensus protocol that simultaneously achieves these goals.
\sysname accomplishes this through a combination of the following techniques.
(1) While state-of-the-art asynchronous protocols, such as Tusk~\cite{narwhal-tusk}, operate over a certified Directed Acyclic Graph (DAG) and attempt to commit one leader block every 9 message delays, \sysname utilizes an \emph{uncertified} DAG as its core data structure. This approach eliminates the overhead associated with the reliable broadcast~\cite{cachin2011introduction} of DAG vertices and allows \sysname to commit most blocks with only five message delays, aligning with the theoretical results of Cordial Miners~\cite{cordial-miners}.
(2) \sysname introduces a novel commit rule that enables the commitment of multiple leader blocks in each DAG round while ensuring safety and liveness in the presence of an asynchronous adversary.
(3) \sysname also explores more practical network assumptions and can be parameterized to further enhance average-case performance while maintaining liveness against a classic asynchronous adversary.


We implement \sysname in Rust and show that it can process an impressive $350,000$ transactions per second in geo-distributed environments with $50$ nodes, all while keeping latency below $2$ seconds. Additionally, \sysname can process $100,000$ transactions per second with latency below $1$ second. 
This achievement sets a new record in the realm of asynchronous consensus protocols and was previously only attainable by partially synchronous protocols~\cite{shoal++,mysticeti,shoal}.
We further show that \sysname maintains the same throughput while improving latency over recent state-of-the-art protocols, Tusk~\cite{narwhal-tusk} and Cordial Miners~\cite{cordial-miners} -- for which we provide the first known implementation -- achieving latency reductions of over $70$\% and $30$\%, respectively.


\para{Contributions} This paper makes the following contributions:
\begin{itemize}
  \item We introduce \sysname, the first asynchronous consensus protocol capable of committing with sub-second latency while maintaining high throughput. Notably, \sysname is the first DAG-based asynchronous consensus protocol capable of committing multiple leader blocks in each round.
  \item We provide detailed algorithms and formal security proofs for \sysname, demonstrating its safety and liveness under an asynchronous network model.
  \item We conduct a formal latency analysis of \sysname, evaluating its commit probability under various network conditions.
  \item We present an implementation and evaluation of \sysname, comparing it to other state-of-the-art protocols and demonstrating that \sysname achieves the lowest commit latency among available asynchronous consensus protocols.
\end{itemize}

%% file: sections/overview.tex
\section{System Overview} \label{sec:overview}

We present an overview of \sysname and the settings in which it operates.

\subsection{Threat model, goals, and assumptions} \label{sec:model}
We consider a message-passing system with $n = 3f + 1$ validators processing transactions using the \sysname protocol. An adversary can adaptively corrupt up to $f$ validators, referred to as \emph{Byzantine}, who may deviate arbitrarily from the protocol. The remaining validators, called \emph{honest}, follow the protocol. The adversary is computationally bounded, ensuring that standard cryptographic properties such as the security of hash functions, digital signatures, and other primitives hold. Under these assumptions, \sysname is \emph{safe} as no two correct validators commit different sequences of transactions.
The communication network is asynchronous and messages can be delayed arbitrarily, but messages among honest validators are eventually delivered. Given these conditions \sysname is \emph{live}, meaning honest validators eventually commit transactions. We provide proofs in \Cref{sec:proofs}.

Furthermore, we analyze \sysname under the \textit{random network model}~\cite{narwhal-tusk}, a variant of the asynchronous network model. While the asynchronous model makes the worst-case assumption that the adversary has perpetually full control over the message schedule (\ie the order in which messages are received by honest validators), the random network model assumes that the message schedule is random (we give a more concrete definition in \Cref{sec:dag-structure}). We analyze \sysname with parameters optimized for the random network model,
representing an average-case evaluation. Our empirical results show that this parameterization generally outperforms a version of \sysname configured for maximum resilience against an asynchronous adversary, all while maintaining safety and liveness guarantees.

\sysname solves Byzantine Atomic Broadcast (BAB)~\cite{cristian1995}, enabling validators to reach consensus on a sequence of messages necessary for State Machine Replication (SMR). According to the FLP impossibility result~\cite{fischer1985}, BAB cannot be solved deterministically in an asynchronous setting. To address this, we employ a global perfect coin to introduce randomization, similar to previous work~\cite{blum2020asynchronous,cachin2000random,dag-rider,loss2018combining}. This coin can be constructed using an adaptively secure threshold signature scheme~\cite{bacho2022on,boneh2001short}, with the distributed key setup performed under fully asynchronous conditions~\cite{abraham2023bingo,abraham2023reaching,das2023practical,das2022practical,kokoris2020asynchronous}.

More formally, each validator $v_k$ broadcasts messages by invoking $\bcast{k}{m,q}$, where $m$ is the message and $q \in \mathbb{N}$ is a sequence number. Every validator $v_i$ has an output $\deliver{i}{m, q, v_k}$, where $m$ is the message, $q$ is the sequence number, and $v_k$ is the identity of the validator that initiated the corresponding $\bcast{k}{m,q}$. \sysname implements a BAB protocol guaranteeing the following properties~\cite{dag-rider}:

\begin{itemize}
    \item \textbf{Validity:} If an honest participant $v_k$ calls $\bcast{k}{m,q}$, then every honest participant $v_i$ eventually outputs $\deliver{i}{m, q, v_k}$, with probability $1$.
    \item \textbf{Agreement:} If an honest participant $v_i$ outputs $\deliver{i}{m, q, v_k}$, then every honest participant $v_j$ eventually outputs $\deliver{j}{m, q, v_k}$ with probability $1$.
    \item \textbf{Integrity:} For each sequence number $q \in \mathbb{N}$ and participant $v_k$, an honest participant $v_i$ outputs $\deliver{i}{m, q, v_k}$ at most once, regardless of $m$.
    \item \textbf{Total Order:} If an honest participant $v_i$ outputs $\deliver{i}{m, q, v_k}$ and

          $\deliver{i}{m', q', v_k'}$ where $q < q'$, all honest participants output

          $\deliver{j}{m, q, v_k}$ before $\deliver{j}{m', q', v_k'}$.
\end{itemize}

\subsection{Intuition behind the \sysname design}\label{sec:intuition}

\sysname builds upon DAG-based consensus protocols that achieve high throughput by processing $O(n)$ blocks per round and fully utilizing network resources. While maintaining these throughput advantages, \sysname focuses on reducing latency in asynchronous state machine replication. It introduces novel techniques to decrease the number of message delays required for block commitment and explores more practical network assumptions to further improve average-case performance.

State-of-the-art asynchronous protocols, such as Tusk~\cite{narwhal-tusk}, operate over a certified DAG and try to commit one leader block every three certified rounds, requiring three message delays to certify each round. This results in at least nine message delays. To reduce latency, \sysname operates over an uncertified DAG by forgoing the reliable broadcast~\cite{cachin2011introduction} of DAG vertices, committing most blocks with only five message delays which matches the theoretical results of Cordial Miners~\cite{cordial-miners}. This approach significantly reduces both bandwidth and compute cost, as validators send their blocks to every other validator only once per round, and they avoid the need to verify cryptographic certificates resulting from consistent broadcast.

This, however, creates the first challenge (\textbf{Challenge 1}): handling equivocations practically. Unlike certified DAG protocols~\cite{narwhal-tusk,dumbo-ng,dag-rider,bullshark,dispersedledger}, \sysname cannot rely on certificates to prevent equivocations, necessitating the design of a novel commit rule to tolerate them. Cordial Miners~\cite{cordial-miners} also face this challenge, but they address it by eventually excluding Byzantine validators that provably equivocate, which can take a long time in asynchrony.

While having five rounds between leaders provides a good probability of committing in asynchronous conditions, it also results in relatively high latency, which is not necessary for ensuring safety.
Although it can be shown that we can implement a commit rule that operates in just three message delays (see \Cref{sec:proofs}), this approach would not work in our network models as it would (1) introduce significant latency variance and (2) lose liveness in the asynchronous model.
Instead, we focus on addressing (\textbf{Challenge 2}): developing a commit rule that effectively reduces average-case latency without sacrificing worst-case liveness. We find that it is possible to reduce the number of rounds to four, achieving a balance between average-case latency in random network conditions and worst-case latency in the classic asynchronous model.

Even with this enhancement, committing only once every four message delays still results in significant latency variance for transactions that are not part of a committed leader block. A primary goal for \sysname is to commit multiple blocks in each round, which would help ensure that the system's tail latency aligns more closely with the four-message delay.
To achieve this, we need to address (\textbf{Challenge 3}): commit every block directly without relying on a sufficient round difference between leader blocks.
If \sysname were to adopt a traditional recursive commit rule~\cite{narwhal-tusk,cordial-miners}, which mandates that each leader block always references all previous leader blocks in their causal history, it would at best be able to commit once every four rounds. However, \sysname recognizes that this causal reference is only necessary when there is no sufficient evidence to directly commit a block, which is not the typical case (\Cref{sec:evaluation}).
This insight indicates that the recursive commit rule used in prior research is overly conservative in its approach to skipping blocks, leading to unnecessary delays, particularly during benign node crashes, which are immediately identifiable. To resolve this issue, we propose a new commit rule capable to promptly determine for each block whether it can be committed or discarded as soon as that decision is evident.

\Cref{sec:decision-rule} presents the \sysname commit rule that addresses these challenges. As a result, \sysname is the first BFT consensus protocol capable of committing multiple blocks per round in the average case, while ensuring both safety and liveness in the asynchronous and random network models.

\subsection{Structure of the \sysname DAG}\label{sec:dag-structure}

We present the structure of the \sysname DAG, building an uncertified DAG that offers similar guarantees to a certified DAG, as shown in related work~\cite{mysticeti,narwhal-tusk,cordial-miners}.

The \sysname protocol operates in a sequence of logical \emph{rounds}. In each round, every honest validator proposes a unique signed \emph{block}, while Byzantine validators may attempt to equivocate by sending multiple blocks or none at all. During a round, validators receive transactions from users and blocks from other validators, which they refer into their proposed blocks. A block includes hash references to blocks from prior rounds, starting with their most recent block, and adds \emph{fresh transactions} not yet included in preceding blocks. Once a block references at least $2f+1$ blocks from the previous round, the validator signs it and broadcasts it.
Clients send transactions to a validator, who adds them to their blocks. If a transaction does not finalize quickly enough, the client sends it to a different validator.

\para{Block creation and validation}
A block must include at least the following elements: (1) the author $A$ of the block and their signature on the block contents; (2) a round number $R$; (3) a list of transactions; (4) at least $2f+1$ distinct hashes of valid blocks from the previous round $R-1$, along with potentially others from prior rounds; and (5) a share of a global perfect coin. As already mentioned the coin can be reconstructed from any $2f+1$ shares.

A block is \emph{valid} if: (1) the signature is valid and
the author
$A$ is part of the validator set; (2) all hashes point to distinct valid blocks from previous rounds, and the sequence of past blocks includes $2f+1$ blocks from the previous round $R-1$; and (3) the share of the global perfect coin is valid\footnote{
    Each individual share of the coin can be independently verified if the coin is implemented through a threshold signature.
}.
Honest validators only include valid blocks into their DAG and discard invalid ones. Furthermore, honest validators only include hashes of blocks once they have downloaded their entire causal history, ensuring that they have successfully validated the block's causal history.

\para{Rounds and waves}
\Cref{fig:rounds} (left) illustrates an example of a \sysname DAG with four validators, $(v_0, v_1, v_2, v_3)$ when parametrized to commit in $5$ rounds.
For the practically efficient $4$-round \sysname, the second \boost round is omitted.

\begin{figure}[t]
    \centering
    \vskip -1em
    \includegraphics[width=\textwidth]{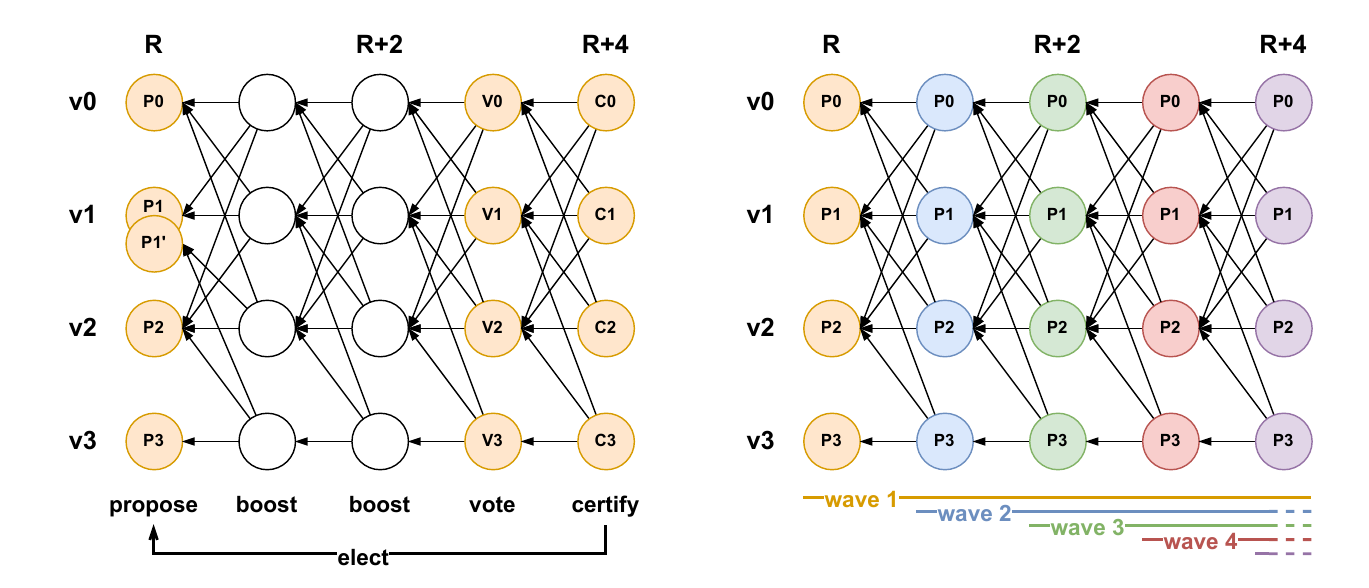}
    \caption{
        The structure of the \sysname DAG. Left: The structure of a wave, consisting of 5 rounds (\propose, \boost, \boost, \vote, \certify). Right: Waves patterns in the \sysname protocol (each round starts a new overlapping wave).
    }
    \vskip -1em
    \label{fig:rounds}
\end{figure}

In its 5-rounds configuration, \sysname defines a \emph{wave} of 5 rounds for every block. The first round (\propose) includes the blocks that the wave attempts to commit ($P_0$, $P_1$, $P_2$, $P_3$) and the equivocating block $P_1'$.
The second and third rounds (\boost) act as a buffer, helping to propagate these blocks to as many validators as possible.
In the fourth round (\vote), every block serves as a \emph{vote} for the first block of the \propose that it encounters when performing a depth-first search following the block hash references. In the example shown in this figure, blocks $V_0$, $V_1$, and $V_2$ are votes for $P_0$, $P_1$, $P_2$ (but not for $P_1'$ and $P_3$), while block $V_3$ is a vote for $P_0$, $P_1$, $P_2$, and $P_3$ (but not for $P_1'$). The procedure $\Call{IsVote}{\cdot}$ of \Cref{alg:helper} (\Cref{sec:algorithms}) formally defines a vote.
The fifth round (\certify) reveals which blocks from the \propose round have been implicitly certified. A block from the \propose round is considered \emph{certified} or \emph{has a certificate} if a block from the \certify round contains in its causal history at least $2f+1$ blocks from the \vote round that are a vote for the block. In this example, blocks $C_0$, $C_1$, $C_2$, and $C_3$ serve as certificates for $P_0$, $P_1$, and $P_2$. This round also opens the global perfect coin, which the decision rule (\Cref{sec:decision-rule}) uses to circumvent the FLP result and to commit blocks under asynchrony by electing some of the \propose blocks as \emph{leaders}. Similar to related work~\cite{narwhal-tusk,dag-rider} this strategy selects leaders ``after the fact'' to deter a network adversary from strategically delaying leader blocks so that they are not referenced by blocks of the \vote round.

As illustrated in \Cref{fig:rounds} (right), \sysname initiates a new wave every round. The rounds of each wave follow a consistent pattern: \propose round: $R$, \boost round: $R+1$, \boost round: $R+2$, \vote round: $R+3$, and \certify round: $R+4$. This pattern repeats continuously, with each new round starting a fresh wave. \Cref{alg:decider} of \Cref{sec:algorithms} formally defines a wave.

\para{Random network model}
We analyze \sysname in the standard asynchronous network model, as well as in the more practical~\cite{mysticeti} random network model~\cite{narwhal-tusk}. In the asynchronous model, the adversary chooses which blocks are received by each honest validator at each round. In contrast, the random network model assumes that at each round $R+1$, an honest validator receives and references valid round-$R$ blocks from a \textit{uniformly random} subset of $2f+1$ validators. \Cref{sec:proofs} provides further details and analyses the commit probability of \sysname in both models.

%% file: sections/protocol.tex
\section{The \sysname Protocol} \label{sec:protocol}
We present \sysname configured with a wavelength of 5 rounds. A configuration of \sysname with a wavelength of 4 rounds operates similarly, but omits one \boost round, and addresses \textbf{challenge 2} of \Cref{sec:overview} as we empirically show in \Cref{sec:evaluation}.

\subsection{Proposers and anchors} \label{sec:proposers}

\sysname leverages a perfect global coin to define several \emph{leader slots} per round. A leader slot is a tuple (validator, round) and can be either empty or contain the validator's proposal for the respective round. If the validator is Byzantine, the slot may also contain more than one (equivocating) block. In line with related work~\cite{mysticeti}, the slot can assume one of three states: \scommit, \sskip, or \sundecided. All slots are initially set to \sundecided and the goal of the protocol is to classify them as \scommit or \sskip. The number of leader slots instantiated per round and the number of boost rounds can be configured (\Cref{sec:evaluation} explores different configurations).

\subsection{The \sysname decision rule} \label{sec:decision-rule}
We present the decision rule of \sysname leveraging an example protocol run. \Cref{sec:algorithms} provides detailed algorithms and \Cref{sec:example} provides a complete step-by-step protocol execution. \Cref{fig:example} illustrates an example of a local view of a \sysname validator, in a system with four validators, $(v_0, v_1, v_2, v_3)$ and parameterized with two leader slots per round. In this example, we refer to blocks using the notation $B_{(v_i, R)}$,
where $v_i$ is the issuing validator and $R$ is the block's round.

\begin{figure}[t]
    \centering
    \vskip -1em
    \includegraphics[width=\textwidth]{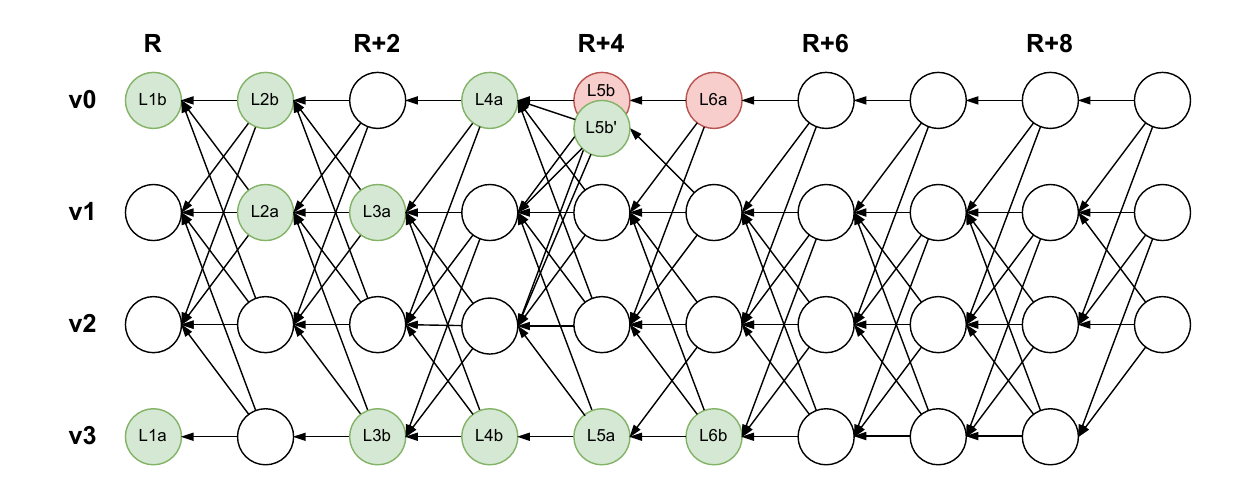}
    \caption{Example execution with 4 validators, wave length of 5 rounds and 2 leader slots per round.}
    \label{fig:example}
    \vskip -1em
\end{figure}

All proposer slots are initially in the \sundecided state. The validator holds the portion of the DAG depicted in \Cref{fig:example} and attempts to classify as many blocks in the leader slots as possible as either \scommit or \sskip.

\para{Step 1: Determine the leader slots}
The validator begins by reconstructing the global perfect coin to determine the leader slots for each round. As shown in \Cref{fig:rounds} (Left), the coin shares embedded in round $R+4$\footnote{
    Or $R+3$ when \sysname is configured with a wave length of 4 rounds.
}
(the \certify round) deterministically establish the leader slots for round $R$ (the \propose round).

In this example, the validator reconstructs the coin from any set of $2f+1$ blocks from round $R+4$ of a wave, then uses it as a seed to deterministically select two leader slots for round $R$: $L_{1a}$ and $L_{1b}$, as illustrated in \Cref{fig:example}. The coin also imposes an order between these two slots: by convention, $L_{1a}$ is the \emph{first} leader slot and $L_{1b}$ is the \emph{second} leader slot of round $R$. The validator repeats this process for every subsequent wave, determining leader slots $L_{2a}$ and $L_{2b}$ from the coin shares in round $R+5$, $L_{3a}$ and $L_{3b}$ from those in round $R+6$, and so on. The validator then sorts these leader slots in descending order: $[L_{6b}, L_{6a}, L_{5b}, L_{5a}, L_{4b}, L_{4a}, L_{3b}, L_{3a}, L_{2b}, L_{2a}, L_{1b}, L_{1a}]$.

This mechanism of determining multiple, potentially empty, leader slots from a global perfect coin is the first step towards addressing \textbf{challenge 3} (\Cref{sec:overview}). Even if validators have different views of the DAG, they will still deterministically try to decide the same leader slots, in the same order, for a given round---regardless of whether they have a block for that slot in memory. This enables \sysname to achieve low latency by electing more than one leader per round and using these slots to order the blocks they causally refer, as described below.

\para{Step 2: Direct decision rule}
The validator attempts to classify each slot, even those without a block as either \scommit or \sskip. To do so, the validator processes each slot individually, starting with the highest ($L_{6b}$), applying the \sysname \emph{direct decision rule}. The validator classifies a block $B$ in a slot as \sskip if it observes $2f+1$ blocks from the subsequent \vote round that do not encounter $B$ when performing a depth-first search following the blocks' references, and as \scommit if it observes $2f+1$ \emph{certificates} over it. As discussed in \Cref{sec:overview}, a certificate over a block $B$ is a block from the \certify round that references at least $2f+1$ blocks from the \vote round, each of which encounter $B$ when performing a depth-first search starting at the voting block. Otherwise, the validator leaves the slot as \sundecided (for now).

In this example, the validator targets $L_{6b}$ first.
It observes that
$B_{(v_0, R+9)}$, $B_{(v_1, R+9)}$, and $B_{(v_2, R+9)}$ are certificates
for $L_{6b}$. Therefore, it classifies $L_{6b}$ as \scommit.
\Cref{sec:evaluation} shows that this scenario is the most common
(in the absence of an asynchronous adversary) and results in the
lowest latency. The validator then
targets
$L_{6a}$ and observes that
$B_{(v_1, R+8)}$, $B_{(v_2, R+8)}$, and $B_{(v_3, R+8)}$ do not vote for it.
Therefore, it classifies $L_{6a}$ as \sskip. The presence of $2f+1$ blocks
from the \vote round that do not vote for a block ensures that it will never
be certified, and will thus never be committed by other validators with a
potentially different local view of the DAG. \Cref{sec:evaluation} shows
that this rule allows \sysname to promptly skip (benign) crashed leaders
to minimize their impact on the protocol's performance.

Malicious validators may attempt to equivocate by creating multiple blocks for the same slot, such as $L_{5b}$ and $L_{5b}'$ in this example. However, the direct decision rule ensures that at most one of these blocks will be classified as \scommit, while the others will be classified as \sskip. In this example, the block $B_{(v_0, R+7)}$ is a vote for $L_{5b}$ (and not for $L_{5b}'$) as it is the first block of the slot encountered when performing a depth-first search starting at $B_{(v_0, R+7)}$ and recursively following all blocks in the sequence of block hashes. Conversely, $B_{(v_1, R+7)}$, $B_{(v_2, R+7)}$, and $B_{(v_3, R+7)}$ are votes for $L_{5b'}$.

This strategy addresses \textbf{challenge 1}. Even though Byzantine validators might equivocate by creating multiple blocks per slot, the causal references defined by the DAG allow the validator to interpret blocks from the \certify round as certificates for blocks from the \propose round. Coupled with the rule that honest validators author at most one block per round, this ensures that at most one block per slot receives a certificate, while all possible other equivocating blocks are skipped. In essence, \sysname embeds the execution of a Byzantine consistent broadcast~\cite{cachin2011introduction} into the DAG.

\para{Step 3: Indirect decision rule}
In the (rare) case where the direct decision rule cannot classify a slot, the validator uses the \sysname \emph{indirect decision rule}. This rule looks at future slots to decide about the current one. First, it finds an \emph{anchor}. This is the first block of the next wave (that is, the earliest slot with a round number
$R' > R + 4$) that is either still classified as \sundecided or already classified as \scommit. If the anchor is \sundecided, the validator marks the current slot as \sundecided. If the anchor is \scommit, the validator checks if it references at least one certificate over the current slot. If it does, the validator marks the current slot as \scommit. If it does not, the validator marks the current slot as \sskip. \Cref{sec:proofs} shows the direct and indirect decision rules are consistent, namely if one validator direct commits a block no honest validators will indirect skip it (and vice versa).

In this example, the validator fails to classify $L_{1a}$ using the direct decision rule as there is only one certificate for $L_{1a}$ and thus
searches for its anchor. Since $L_{6a}$ has been classified as \sskip, it cannot serve as an anchor; therefore, $L_{6b}$ becomes the anchor for $L_{1a}$. Given that block $B_{(v_3, R+4)}$, which serves as a certificate for $L_{1a}$, is referenced in $L_{6b}$'s causal history, the validator classifies $L_{1a}$ as \scommit.

This rule is the last step to solving \textbf{challenge 3}. It allows the validator to indirectly decide on a block by leveraging the earliest anchors rather than waiting for the next leader slot which may come much later. This enables \sysname to eliminate the need for non-leader blocks between leader slots, achieving low latency by electing leader slots in every round.

\para{Step 4: Leader slots sequence}
After processing all slots, the validator derives an ordered sequence of the blocks contained in the leader slots.
It then iterates over this sequence, committing all slots marked as \scommit and skipping all slots marked as \sskip.
This process continues until the validator encounters the first \sundecided slot. As demonstrated in \Cref{sec:proofs},
this commit sequence is safe, and eventually, all slots will be classified as either \scommit or \sskip.

In the example shown in \Cref{fig:example}, the leader sequence output by the validator is [$L_{1a}$, $L_{1b}$, $L_{2a}$, $L_{2b}$, $L_{3a}$, $L_{3b}$, $L_{4a}$, $L_{4b}$, $L_{5a}$, $L_{5b}'$, $L_{6b}$]
. \Cref{sec:example} provides a detailed walkthrough of the decision rule applied to the example DAG in \Cref{fig:example}, guiding the reader step-by-step through deriving this commit sequence.

\para{Step 5: Commit sequence}
Following the approach introduced by DagRider~\cite{dag-rider}, the validator linearizes the blocks within the sub-DAG defined by each leader block by performing a depth-first search. If a block has already been linearized by a previous leader slot, it is not re-linearized. The validator processes each leader slot sequentially, ensuring that all blocks are included in the final commit sequence in the correct order, according to their causal dependencies. The procedure $\Call{LinearizeSubDags}{\cdot}$ of \Cref{alg:helper} (\Cref{sec:algorithms}) formally describes this process.

In this example, $L_{1a}$ and $L_{1b}$ do not define any sub-DAG (the example begins at round $R$) and are thus directly added to the commit sequence. Next, $L_{2a}$ defines the sub-DAG $\{ L_{1b}, B_{(v_1, R)}, B_{(v_2, R)}, L_{2a} \}$, which is linearized as [$B_{(v_1, R)}$, $B_{(v_2, R)}$, $L_{2a}$] since $L_{1b}$ is already part of the commit sequence.
The validator continues this process for each leader in the sequence, linearizing the sub-DAGs defined by $L_{3b}$, then $L_{3a}$ and so forth following the procedure above. The final commit sequence is
    [
        $L_{1a}$, $L_{1b}$,
        $B_{(v_1, R)}$, $B_{(v_2, R)}$, $L_{2a}$, $L_{2b}$,
        $B_{(v_2, R+1)}$, $L_{3a}$, $B_{(v_3, R+1)}$, $L_{3b}$,
        $B_{(v_0, R+2)}$, $B_{(v_2, R+2)}$, $L_{4a}$, $L_{4b}$,
        $B_{(v_1, R+3)}$, $B_{(v_2, R+3)}$, $L_{5a}$, $L_{5b}'$,
        $B_{(v_1, R+4)}$, $B_{(v_2, R+4)}$, $L_{6b}$
    ].

%% file: sections/implementation.tex
\section{Implementation} \label{sec:implementation}

We implemented a networked, multi-core \sysname validator in Rust by forking the
\mysticeti codebase~\cite{mysticeti-code}, consisting of about $14,000$ LOC.
Our implementation utilizes \texttt{tokio}~\cite{tokio} for asynchronous networking and employs raw TCP sockets for communication. We rely on \texttt{ed25519-consensus}~\cite{ed25519-consensus} for asymmetric cryptography and \texttt{blake2}~\cite{rustcrypto-hashes} for cryptographic hashing. To ensure data persistence and crash recovery, we implemented a Write-Ahead Log (WAL) tailored to the unique requirements of our consensus protocol.
Furthermore, we implemented Cordial Miners~\cite{cordial-miners}, a state-of-the-art DAG-based asynchronous consensus protocol, using the same system components.
This enabled us to to perform a comparative evaluation with \sysname, see \Cref{sec:evaluation}.
Since the Cordial Miners paper lacks both implementation and evaluation, we believe our implementation and evaluation are additional contributions of our work.
We are open-sourcing both our implementations of \sysname and Cordial Miners, along with our orchestration tools, to ensure reproducibility of our results\footnote{\codelink}.

%% file: sections/evaluation.tex
\section{Evaluation}\label{sec:evaluation}

We evaluate the throughput and latency of \sysname through experiments conducted on Amazon Web Services (AWS), demonstrating its performance improvements over the state-of-the-art. We evaluate \sysname with different parametrizations, with a wave length of $4$ and $5$ and with different numbers of leaders per round.

We compare \sysname with Tusk~\cite{narwhal-tusk}, as an example of certified DAG-based consensus protocol, and Cordial Miners~\cite{cordial-miners}, as an example of an uncertified DAG-based protocol.
We choose these protocols because, to the best of our knowledge, Tusk has shown the highest throughput among all published and implemented asynchronous BFT protocols when evaluated in a geo-distributed environment.
Cordial Miners, while lacking an implementation and evaluation, theoretically proves excellent latency bounds and is the protocol most similar to \sysname.
We also considered a performance comparison with other recent asynchronous consensus protocols, including Pace~\cite{pace}, Fin~\cite{fin}, ParBFT~\cite{parbft}, and SQ~\cite{sui2023}, but ultimately decided against them. The reasons for this decision is that their implementations are either closed-source, only capable of handling a limited number of block proposals (leading to crashes under sustained load), or unable to operate in a WAN environment (resulting in deadlocks after a few seconds).

Our evaluation particularly aims to demonstrate the following claims:
\begin{enumerate}[label={\bf C\arabic*}]
    \item\label{claim:c1} \sysname has similar throughput and lower latency than the baseline state-of-the-art protocols when operating in ideal conditions.
    \item\label{claim:c2} \sysname scales well by maintaining high throughput and low latency as the number of validators increases.
    \item\label{claim:c3} \sysname has a similar throughput to, and lower latency than, Cordial Miners, when operating in the presence of (benign) crash faults.
    \item\label{claim:c4} \sysname latency decreases when increasing the number of leader slots per round (up to 3 leaders per round).
    \item\label{claim:c5} \sysname parametrized with a wave length of $4$ rounds has lower latency in our geo-replicated network than when configured with a wave length of $5$ rounds.
\end{enumerate}
Note that evaluating the performance of BFT protocols in the presence of Byzantine faults is an open research question~\cite{twins}, and state-of-the-art evidence relies on formal proofs of safety and liveness (presented in \Cref{sec:proofs}).
While there is a need to robustly tolerate Byzantine faults, we note that they are rare in observed delegated proof-of-stake blockchains, as compared to crash faults which occur commonly~\cite{mysticeti}.

\subsection{Experimental Setup}

We deploy \sysname on AWS, using \texttt{m5d.8xlarge} instances across $5$ different AWS regions:
Ohio (us-east-2), Oregon (us-west-2), Cape Town (af-south-1), Hong Kong (ap-east-1), and Milan (eu-south-1).
Validators are distributed across those regions as equally as possible.
Each machine provides $10$\,Gbps of bandwidth,
$32$ virtual CPUs (16 physical cores) on a $3.1$\,GHz Intel Xeon Skylake 8175M,
$128$\,GB memory,
and runs Linux Ubuntu server $22.04$.
We select these machines because they provide decent performance and are in the price range of ``commodity servers''.

In the following, \emph{latency} refers to the time elapsed from the moment a client submits a transaction to when it is committed by the validators,
and \emph{throughput} refers to the number of transactions committed per second.
Each data point is the average latency of $3$ runs and the error bars represent one standard deviation (error bars are sometimes too small to be visible on the graph).
We instantiate several geo-distributed benchmark clients within each validator submitting transactions in an open loop model, at a fixed rate.
We experimentally increase the load of transactions sent to the systems, and record the throughput and latency of commits.
As a result, all plots illustrate the steady-state latency of all systems under low load, as well as the maximal throughput they can provide after which latency grows quickly.
Transactions in the benchmarks are arbitrary and contain $512$ bytes.
Unless stated otherwise, we configure \sysname with $2$ leaders per round.
In the following graphs, we refer to \sysname with a wave length of $5$ as \sysname-5 and \sysname with a wave length of $4$ as \sysname-4.

\subsection{Benchmark under ideal conditions}

We assess the performance of \sysname under normal, failure-free conditions in a wide-area network (WAN) environment.
\Cref{fig:best-case-performance} presents the performance results of \sysname in a geo-replicated setting,
comparing both a small committee of $10$ validators and a large committee of $50$ validators.

\begin{figure}[t]
    \vskip -1em
    \centering
    \includegraphics[width=\linewidth]{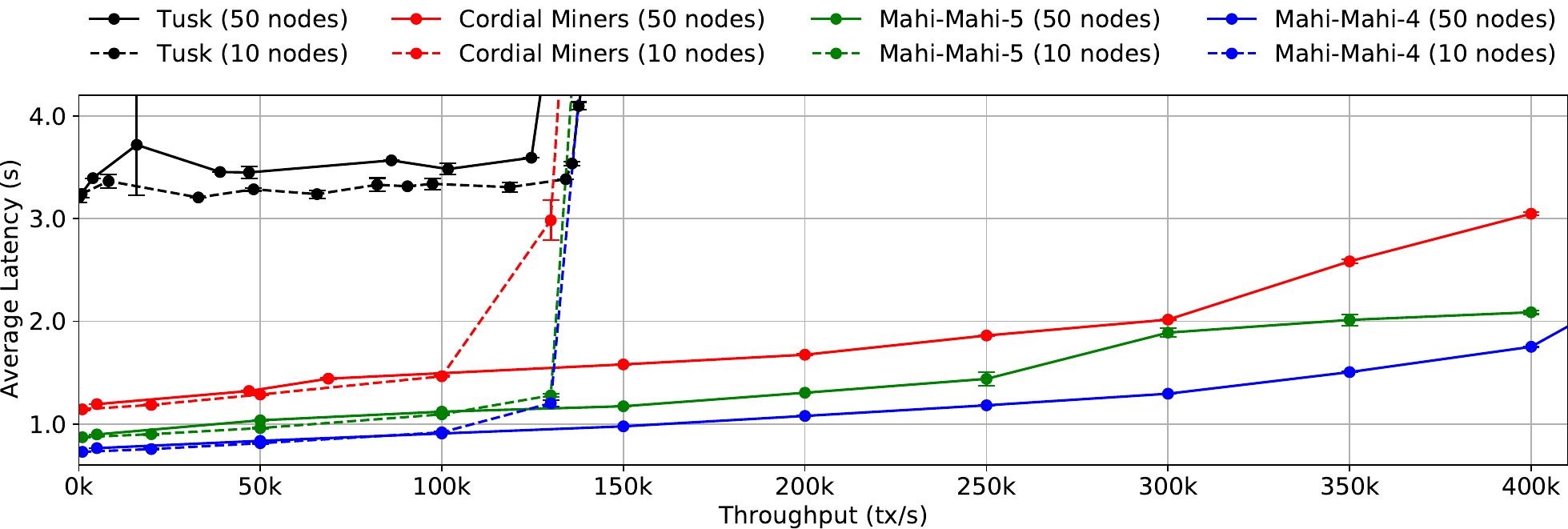}
    \caption{
        Comparative throughput-latency performance of \sysname, Tusk, and Cordial Miners. WAN measurements with $10$ and $50$ validators. No validator faults. $512$B transaction size.
    }
    \label{fig:best-case-performance}
    \vskip -1em
\end{figure}

For a small committee of $10$ nodes, all three systems---Tusk, Cordial Miners, and \sysname---reach a peak throughput of approximately $100$k-$130$k transactions per second (tx/s).
However, their latencies vary significantly.
Tusk and Cordial Miners achieve average latencies of
$3.5$s and
$1.5$s, respectively.
In contrast,
\sysname configured with a wave length $5$ has a latency of $1.1$s,
representing a substantial reduction of
$68$\% compared to Tusk and
$27$\% compared to Cordial Miners.
\sysname with wave length $4$ has a latency of $0.9$s,
representing a substantial reduction of
$74$\% compared to Tusk and
$40$\% compared to Cordial Miners.
Tusk's higher latency stems from its certified DAG architecture,
requiring at least $9$ network messages to commit a block.
While Cordial Miners bypasses DAG certification,
it can only commit one leader every $5$ rounds.
In contrast, \sysname operating with wave length 5 consistently commits multiple blocks.
\sysname operating with wave length $4$ further reduces latency as it commits blocks after $4$ message delays.
These results validate our claim~\ref{claim:c1}.

For a large committee of 50 nodes,
\Cref{fig:best-case-performance} shows that the throughput of Cordial Miners and \sysname exceeds 350,000 transactions per second (tx/s), while Tusk's throughput remains around 125,000 tx/s.
This perhaps surprising increase in throughput occurs
because our \sysname's validator implementation is optimized for large networks and
does not fully utilize all available resources (network, disk, CPU) when deployed with smaller committees.
Consequently, adding more validators improves resource multiplexing, boosting \sysname's performance.
Additionally, as the committee size grows, the number of blocks per round increases, thus a larger
number of blocks are included in the causal history of elected leader blocks, without incurring additional network hops.
Unlike Tusk, both Cordial Miners and \sysname experience no significant CPU overhead as the committee size increases, and bandwidth does not become a bottleneck at these throughput levels.
However, we do not expect further throughput gains by increasing the committee size beyond 50 nodes (such experiments would be prohibitively expensive).
As expected, Cordial Miners and \sysname share nearly identical throughput since both rely on the same DAG implementation, and throughput is determined by the efficiency of the DAG propagation layer.

In terms of latency, Tusk and Cordial Miners achieve average latency of
3.5s and
2.6s, respectively.
\sysname parametrized with a wave length of 5 has a latency of
2s (at 350,000 tx/s), which is a
42\% reduction compared to Tusk and a
23\% reduction compared to Cordial Miners.
\sysname with a wave length 4 has a latency of 1.5s, which is a
57\% reduction compared to Tusk and a
42\% reduction compared to Cordial miners.
These results validate our claim~\ref{claim:c2}. Comparing the two versions of \sysname in those two experiments also validates our claim~\ref{claim:c5}.

\subsection{Performance under faults}
\Cref{fig:crash} depicts the performance of all systems
when a committee of 10 validators suffers 3 crash-faults (the maximum that can be tolerated for this committee size).

\begin{figure}[t]
    \vskip -1em
    \begin{minipage}{0.48\textwidth}
        \centering
        \includegraphics[width=\textwidth]{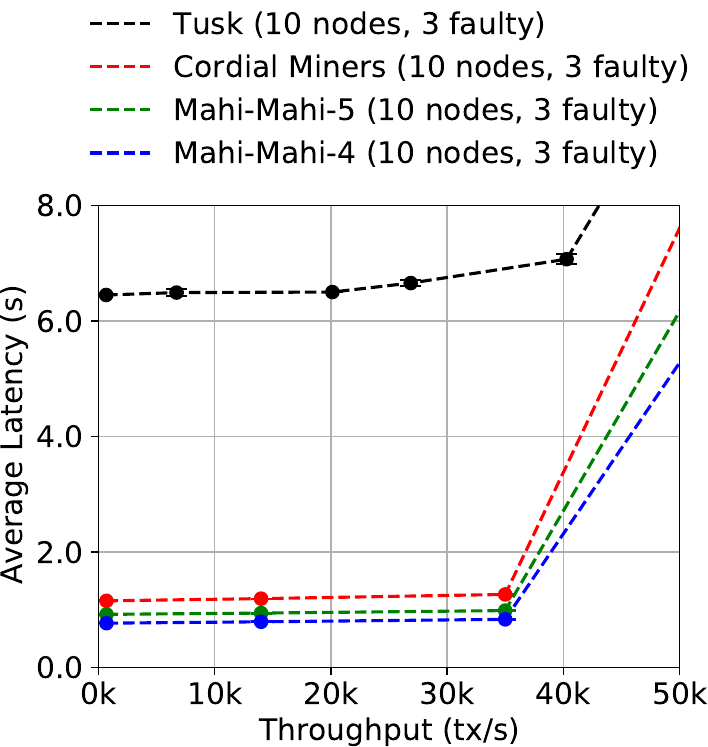}
        \caption{
            Comparative throughput-latency of \sysname, Tusk, and Cordial Miners. WAN measurements with 10 validators. Three faults. 512B transaction size.
        }
        \label{fig:crash}
    \end{minipage}
    \hfill
    \begin{minipage}{0.48\textwidth}
        \centering
        \includegraphics[width=\textwidth]{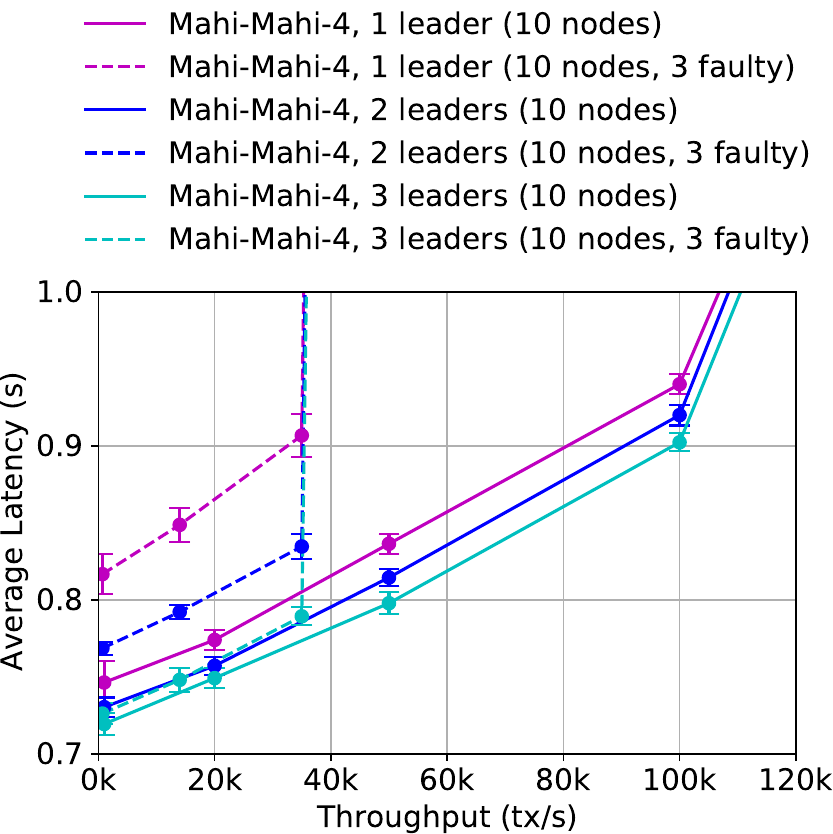}
        \caption{
            Impact of the number of leaders per round in \sysname. WAN measurements with 10 validators. Zero and three faults. 512B transaction size.
        }
        \label{fig:leaders}
    \end{minipage}
    \vskip -1em
\end{figure}

We observe that all three systems achieve a throughput of approximately $35,000$-$40,000$ tx/s.
Tusk and Cordial Miners record a latency of around $7$s and $1.7$s, respectively.
\sysname records a latency of $0.95$s and $0.85$s when running with a wave length $5$ and $4$, respectively.
Despite the presence of faulty validators, the DAG continues to collect and disseminate transactions without significant impact.
The reduction in throughput seen in \Cref{fig:crash}, compared to \Cref{fig:best-case-performance}, can be attributed to two primary factors:
(1) the loss of capacity due to faulty validators, and
(2) the higher frequency of missing elected leader blocks, which leads to increased commit delays.
As expected, the latency advantage of \sysname over Cordial Miners narrows when operating with a high number of faulty leaders, which cause head-of-line blocking and prevent the protocol from promptly committing future leaders.
However, \sysname still maintains a latency advantage of approximately
$50$\% over Cordial Miners, thanks to its direct skip rule (\Cref{sec:protocol}),
which allows \sysname to bypass faulty leaders roughly 2 rounds earlier than Cordial Miners.
Thus, our claim~\ref{claim:c3} holds.

\subsection{Impact of the number of leader slots per round}

Finally, we assess the impact of multiple leaders on \sysname's performance.
We experiment with \sysname parametrized with wave lengths of 4 and 5.
\Cref{fig:leaders} illustrates how \sysname configured with a wave length of 4 rounds performs with 1, 2, and 3 leaders under both normal conditions and scenarios involving 3 crash faults.
\Cref{fig:leaders-w-5} in \Cref{sec:eval-impact-multi-leader} shows the same experiment with wave length 5.
We observe a notable reduction in average latency as the number of leaders increases.
Specifically, when the number of leaders rises from 1 to 3, \sysname's average latency decreases by approximately
40ms in the ideal scenario,
and by approximately
100ms in the crash failure scenario.
This improvement arises because having more leaders per round increases the number of blocks committed directly by leaders, rather than through the causal history of previous leader blocks.
These findings validate our claim~\ref{claim:c4}.
Increasing the number of leaders beyond 3 did not further decrease latency.
This is due to the higher likelihood of failing to commit via the direct decision rule, which may cause head-of-line blocking and delays the commitment of future leaders.

%% file: sections/related-work.tex
\section{Related Work}\label{sec:related-work}
We compare \sysname with several categories of related work.

\para{Uncertified DAG-based consensus protocols}
The system most similar to \sysname is Cordial Miners~\cite{cordial-miners}. Like \sysname, Cordial Miners operates over an uncertified DAG, where each vertex represents a block that is disseminated with best-effort to all peers~\cite{threshold-clock}. The primary distinction between the two lies in their commit rules. Cordial Miners can commit at most one leader block every five rounds, which leads to significantly higher latency for transactions not included in that leader block. In contrast, \sysname's commit rule allows for a configurable number of blocks to be committed in each round, increasing the number of blocks committed per round and reducing the latency for most transactions. \sysname commits more blocks directly through leaders, rather than relying on the causal history of previous leader blocks. Additionally, Cordial Miners does not provide an implementation or evaluation.

Mysticeti~\cite{mysticeti} is a recent protocol that, like \sysname, operates over an uncertified DAG but in a partially synchronous setting. Mysticeti takes advantage of synchronous periods in the network to commit blocks in three rounds, and like \sysname, it can commit blocks every round. However, unlike \sysname and other asynchronous protocols, Mysticeti completely loses liveness when the network is not synchronous. To maintain liveness in asynchronous conditions, \sysname interprets the DAG differently from Mysticeti. Specifically, \sysname incorporates a global perfect coin into the protocol and modifies the role of several DAG
rounds
to ensure that an asynchronous adversary cannot indefinitely manipulate message schedules to prevent block certificates from forming---an issue that can easily arise in Mysticeti~\cite{consensus-dos}.

\para{Certified DAG-based consensus protocols}
DAG-Rider~\cite{dag-rider}, Tusk~\cite{narwhal-tusk}, and Dumbo-NG~\cite{dumbo-ng} are popular
asynchronous certified DAG-based consensus protocols that use reliable or consistent broadcast to explicitly certify every DAG vertex~\cite{sok-dag}. This approach introduces 3 message delays per DAG round but simplifies the commit rule by ensuring that equivocating DAG vertices never occur. However, this method results in significantly higher latency compared to \sysname. For instance, DAG-Rider requires at least 12 messages to commit a block, while Tusk and Dumbo-NG require 9 messages. By contrast, \sysname can commit in just 4 or 5 message delays when respectively configured with a wave length of 4 and 5. Also, certified DAGs have higher bandwidth and CPU requirements, as validators must disseminate, receive, and verify the cryptographic certificates generated by consistent broadcast. As shown in \Cref{sec:evaluation}, these factors lead to up to 70\% higher latency in comparison to \sysname.

Sailfish~\cite{sailfish}, BBCA-Chain~\cite{bbca-chain}, Fino~\cite{fino}, Shoal~\cite{shoal}, and Shoal++\cite{shoal++} build on the partially synchronous version of Bullshark~\cite{bullshark-partial-sync} through various improvements, including the ability to commit more blocks per round and a relaxation of DAG certification requirements. However, these protocols are limited to partially synchronous environments and, unlike \sysname, they lose liveness in asynchronous conditions.

\para{Linear-chain protocols}
Linear-chain asynchronous protocols such as Das \etal~\cite{das2024}, Pace~\cite{pace}, FIN~\cite{fin}, and SQ~\cite{sui2023} do not leverage an underlying DAG structure. They instead rely on explicit Byzantine consistent broadcast~\cite{cachin2011introduction} and a common coin to elect a leader, whereas \sysname incorporates these components implicitly within the DAG. This leader drives the protocol by constructing a linear chain. Consequently, these protocols do not achieve the same level of throughput and robustness as DAG-based systems~\cite{narwhal-tusk}. Their contributions instead lie primarily in their theoretical foundations. For example, Das \etal introduces a protocol that operates without a trusted setup or the need for public-key cryptography; FIN presents the first constant-time asynchronous consensus (ACS) protocol with $O(n^3)$ messages in both information-theoretic and signature-free settings; and SQ reduces this message complexity to $O(n^2)$.


%% file: sections/conclusion.tex
\section{Conclusion}\label{sec:conclusion}

We introduce \sysname, an asynchronous consensus protocol achieving a new performance milestone: \sysname can process an impressive 350,000 transactions per second in geo-distributed environments with 50 nodes all while keeping latency below 2 seconds, or 100,000 transactions per second with sub-second latency---an achievement that sets a new record in the realm of asynchronous consensus protocols and that was only thought possible for partially-synchronous protocols. The exceptional performance is made possible through a novel commit rule applied over an uncertified DAG that enables commits of multiple leaders every round. This allows \sysname to inherit the robustness and throughput inherent in DAG-based protocols, all while establishing a new standard for the latency of asynchronous consensus protocols.

%% file: acks.tex
\subsubsection*{Acknowledgements.} This work is partially supported by Mysten Labs and the Sui Foundation. This work is also partially supported by the Gates Foundation [INV-057591]; under the grant conditions of the Foundation, a Creative Commons Attribution 4.0 Generic License has already been assigned to the Author's Accepted Manuscript.

%% file: sections/algorithms.tex
\section{\sysname Algorithms} \label{sec:algorithms}

This section presents the algorithms used in \sysname in pseudocode format. If a high-level understanding of \sysname is sufficient, then this section can be skipped.
In particular, \Cref{alg:main} specifies the \sysname main algorithms, \Cref{alg:decider} the \sysname decider instance, and \Cref{alg:helper} contains various DAG helper functions.
As a reminder, \sysname operates with a single type of message: a block whose validity is described in \Cref{sec:dag-structure}. Validators hold these blocks in a data structure called $\store$. To access the block(s) of round $r$ authored by validator $v$ of the DAG, we write $\store[r, v]$. If an equivocation happened at a slot $v$, then $\store[r, v]$ may return multiple blocks. To access all blocks of a given round $r$, we write $\store[r, *]$.

Then entry point is the procedure $\Call{ExtendCommitSequence}{\cdot}$ (Line~\ref{alg:main:extend-sequence} of \Cref{alg:main}), which is called by the application layer to extend the commit sequence. This procedure is idempotent and is called by our implementation (\Cref{sec:implementation}) every time the validator receives a new block. This procedure calls $\Call{TryDecide}{\cdot}$ (Line~\ref{alg:main:try-decide} of \Cref{alg:main}) to classify as many blocks as possible as either \scommit or \sskip. The $\Call{TryDecide}{\cdot}$ procedure iterates over all possible leaders and invokes the decider instance (Line~\ref{alg:main:decider} of \Cref{alg:main}) to classify each leader slot. The decider instance is responsible for determining the leader of a given round, certifying blocks, and classifying leader slots. The decider instance uses various helper functions, such as $\Call{IsVote}{\cdot}$ (Line~\ref{alg:line:isvote} of \Cref{alg:helper}), and $\Call{IsCert}{\cdot}$ (Line~\ref{alg:line:iscert} of \Cref{alg:helper}), that are generic utilities for working with the DAG.

\begin{algorithm}[t]
    \caption{\sysname Main Function}
    \label{alg:main}
    \algsize

    \begin{algorithmic}[1]
        \State \wavelength \Comment{Set to at least $4$ (see \Cref{sec:protocol})}
        \State \numleaders \Comment{See \Cref{sec:evaluation} for details}

        \Statex
        \Statex // Idempotent function called by the application layer to extend the commit sequence.
        \Procedure{ExtendCommitSequence}{$r_{committed}, r_{highest}$} \label{alg:main:extend-sequence}
        \State $L \gets \Call{TryDecide}{r_{committed}, r_{highest}}$ \Comment{See \Cref{alg:main:try-decide} below}
        \State $L_{commit} \gets [\;]$ \Comment{Hold decided leader sequence}
        \For{$\status \in L$}
        \If{$\status = \perp$} \textbf{break} \Comment{Stop at the first \sundecided leader} \EndIf
        \If{$\status = \scommit(b_{leader})$}
        \State $L_{commit} \gets L_{commit} \; || \; b_{leader}$
        \EndIf
        \EndFor
        \State \Return $\Call{LinearizeSubDags}{L_{commit}}$ \Comment{See Line~\ref{alg:helper:linearize} of \Cref{alg:helper}}
        \EndProcedure

        \Statex
        \Statex // Try to decide as many proposals as possible, recursively, starting from the latest proposal.
        \Procedure{TryDecide}{$r_{committed}, r_{highest}$} \label{alg:main:try-decide}
        \State $L \gets [\;]$ \Comment{Hold decision of each leader}
        \For{$r \in [r_{highest} \text{ down to } r_{committed}+1]$}
        \For{$l \in [\numleaders-1 \text{ down to } 0]$} \Comment{Loop over all possible leaders}
        \State $i \gets r \; \% \; \wavelength$
        \State $\decider \gets \textsf{Decider}(\wavelength, i, l)$ \Comment{\Cref{alg:decider}} \label{alg:main:decider}
        \State $w \gets \decider.\Call{WaveNumber}{r}$
        \If{$\decider.\Call{ProposeRound}{w} \neq r$} \textbf{continue} \Comment{Skip if not a leader} \EndIf
        \State $\status \gets \decider.\Call{TryDirectDecide}{w}$ \Comment{Apply direct decision rule}
        \If{$\status = \perp$} \label{alg:line:universal:direct-decide-failed}
        \State $\status \gets \decider.\Call{TryIndirectDecide}{w}$ \Comment{Apply indirect decision rule}
        \EndIf
        \State $L \gets \status \; || \; L$
        \EndFor
        \EndFor
        \State \Return $L$ \Comment{May still contain \sundecided leaders}
        \EndProcedure
    \end{algorithmic}
\end{algorithm}

\begin{algorithm}[!ht]
    \caption{\sysname Decider Instance}
    \label{alg:decider}
    \algsize

    \begin{algorithmic}[1]
        \State \wavelength \Comment{Set to at least $4$ (see \Cref{sec:protocol})}
        \State \waveoffset \Comment{Offset creating overlapping waves (\Cref{sec:overview})}
        \State \leaderoffset \Comment{Each decider operates on a unique leader slot}

        \Statex
        \Procedure{WaveNumber}{$r$}
        \State \Return $(r - \waveoffset) / \wavelength$
        \EndProcedure

        \Statex
        \Procedure{ProposeRound}{$w$}
        \State \Return $w \cdot \wavelength + \waveoffset$ \Comment{See \Cref{fig:rounds}}
        \EndProcedure

        \Statex
        \Procedure{CertifyRound}{$w$}
        \State \Return $w \cdot \wavelength + \wavelength - 1 + \waveoffset$ \Comment{See \Cref{fig:rounds}}
        \EndProcedure

        \Statex
        \Procedure{VoteRound}{$w$}
        \State \Return $\self.\Call{CertifyRound}{w} - 1$ \Comment{See \Cref{fig:rounds}}
        \EndProcedure

        \Statex
        \Procedure{LeaderBlock}{$w$}
        \State $r_{propose}, r_{certify} \gets \self.\Call{ProposeRound}{w}, \self.\Call{CertifyRound}{w}$
        \State $c \gets \Call{CombineCoinShares}{\{ b.share \text{ s.t. } b \in \store[r_{certify}, *] \}}$ \Comment{Common coin}
        \State $l \gets c + \leaderoffset$ \Comment{Modulo committee size}
        \State \Return $\store[r_{propose}, l]$ \Comment{May return more than one block in case of equivocations}
        \EndProcedure

        \Statex
        \Procedure{SkippedLeader}{$w, b_{leader}$}
        \State $r_{vote} \gets \self.\Call{VoteRound}{w}$
        \State \Return $|\{ \neg \Call{IsVote}{b, b_{leader}} \text{ s.t. } b \in \store[r_{vote}, *] \}| \geq 2f+1$
        \EndProcedure

        \Statex
        \Procedure{SupportedLeader}{$w, b_{leader}$}
        \State $r_{certify} \gets \self.\Call{CertifyRound}{w}$
        \State \Return $|\{ \Call{IsCert}{b, b_{leader}} \text{ s.t. } b \in \store[r_{certify}, *] \}| \geq 2f+1$
        \EndProcedure

        \Statex
        \Procedure{TryDirectDecide}{$w$}
        \For{$b_{leader} \in \self.\Call{LeaderBlock}{w}$} \Comment{Loop over equivocations}
        \If{$\self.\Call{SkippedLeader}{w, b_{leader}}$} \Return $\sskip(w)$ \EndIf
        \If{$\self.\Call{SupportedLeader}{w, b_{leader}}$} \Return $\scommit(b_{leader})$ \EndIf
        \EndFor
        \State \Return $\perp$
        \EndProcedure

        \Statex
        \Procedure{TryIndirectDecide}{$w, S$}
        \State $s_{anchor} \gets \text{ find first } s \in S \text{ s.t. } r_{certify} < s.round \wedge s \neq \sskip(w)$
        \If{$s_{anchor} = \scommit(b_{anchor})$}
        \If{$\exists b_{leader} \in \self.\Call{LeaderBlock}{w} \text{ s.t. } \Call{IsCertifiedLink}{b_{anchor}, b_{leader}}$}
        \State \Return $\scommit(b_{leader})$
        \Else
        \State \Return $\sskip(w)$
        \EndIf
        \EndIf
        \State \Return $\perp$ \Comment{The anchor is \sundecided or not found}
        \EndProcedure
    \end{algorithmic}
\end{algorithm}

\begin{algorithm}[!ht]
    \caption{DAG Helper Functions}
    \label{alg:helper}
    \algsize

    \begin{algorithmic}[1]
        \Statex
        \Procedure{IsVote}{$b_{vote}, b_{leader}$} \label{alg:line:isvote}
        \Function{VotedBlock}{$b, id, r$}
        \If{$r \geq b.round$} \Return $\perp$ \EndIf
        \For{$b' \in b.parents$}
        \If{$(b'.author, b'.round) = (id, r)$} \Return $b'$ \EndIf
        \State $res \gets \Call{VotedBlock}{b', id, r}$
        \If{$res \neq \perp$} \Return $res$ \EndIf
        \EndFor
        \State \Return $\perp$
        \EndFunction
        \State $(id, r) \gets (b_{leader}.author, b_{leader}.round)$
        \State \Return $\Call{VotedBlock}{b_{vote}, id, r} = b_{leader}$
        \EndProcedure

        \Statex
        \Procedure{IsCert}{$b_{cert}, b_{leader}$} \label{alg:line:iscert}
        \State $res \gets |\{b \in b_{cert}.parents: \Call{IsVote}{b, b_{leader}}\}|$
        \State \Return $res \geq 2f+1$
        \EndProcedure

        \Statex
        \Procedure{IsLink}{$b_{old}, b_{new}$}
        \State \Return exists a sequence of $k\in\mathbb{N}$ blocks $b_1, \dots, b_k$ s.t. $b_1 = b_{old}, b_k = b_{new}$ and $\forall j \in [2, k]: b_j \in \bigcup_{r \geq 1} \store[r, *] \wedge b_{j-1} \in b_{j}.parents$
        \EndProcedure

        \Statex
        \Procedure{IsCertifiedLink}{$b_{anchor}, b_{leader}$}
        \State $w \gets \Call{WaveNumber}{b_{leader}.round}$
        \State $B \gets \Call{GetDecisionBlocks}{w}$
        \State \Return $\exists b \in B \text{ s.t. } \Call{IsCert}{b, b_{leader}} \wedge \Call{IsLink}{b, b_{anchor}}$
        \EndProcedure

        \Statex
        \Procedure{LinearizeSubDags}{$L$} \label{alg:helper:linearize}
        \State $O \gets [\;]$ \Comment{Hold output sequence}
        \For{$b_{leader} \in L$}
        \State $B \gets \{b \in \bigcup_{r \geq 1} \store[r, *] \text{ s.t. } \Call{IsLink}{b, b_{leader}} \wedge b \notin O \wedge b \text{ not already output } \}$
        \For{$b \in B \text{ in any deterministic order}$}
        \State $O \gets O \; || \; b$
        \EndFor
        \EndFor
        \State \Return $O$
        \EndProcedure
    \end{algorithmic}
\end{algorithm}

%% file: sections/example.tex
\section{Example of \sysname Execution} \label{sec:example}

This section completes \Cref{sec:protocol} by guiding readers through the protocol execution, step-by-step protocol, using the example illustrated in \Cref{fig:example}. This figure showcases \sysname with four validators, labeled as ($v_0$, $v_1$, $v_2$, $v_3$), operating over a wavelength of five rounds, with two leader slots allocated per round.
As mentioned in \Cref{sec:decision-rule}, all proposer slots are initially in the \sundecided state. The validator holds in memory the portion of the DAG depicted in \Cref{fig:example} and attempts to classify as many blocks in the leader slots as possible as either \scommit or \sskip.

\begin{figure}[t]
    \centering
    \begin{subfigure}[t]{\textwidth}
        \centering
        \includegraphics[width=0.75\textwidth]{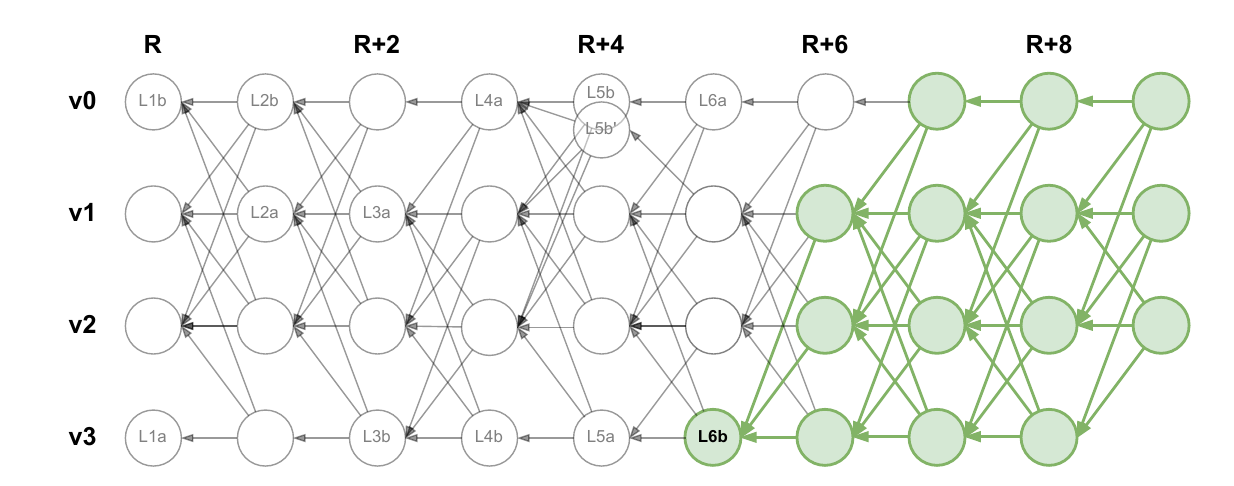}
        \caption{Example of direct decision rule. $L_{6b}$ is classified as \scommit.}
        \label{fig:example-1}
    \end{subfigure}
    \begin{subfigure}[t]{\textwidth}
        \centering
        \includegraphics[width=0.75\textwidth]{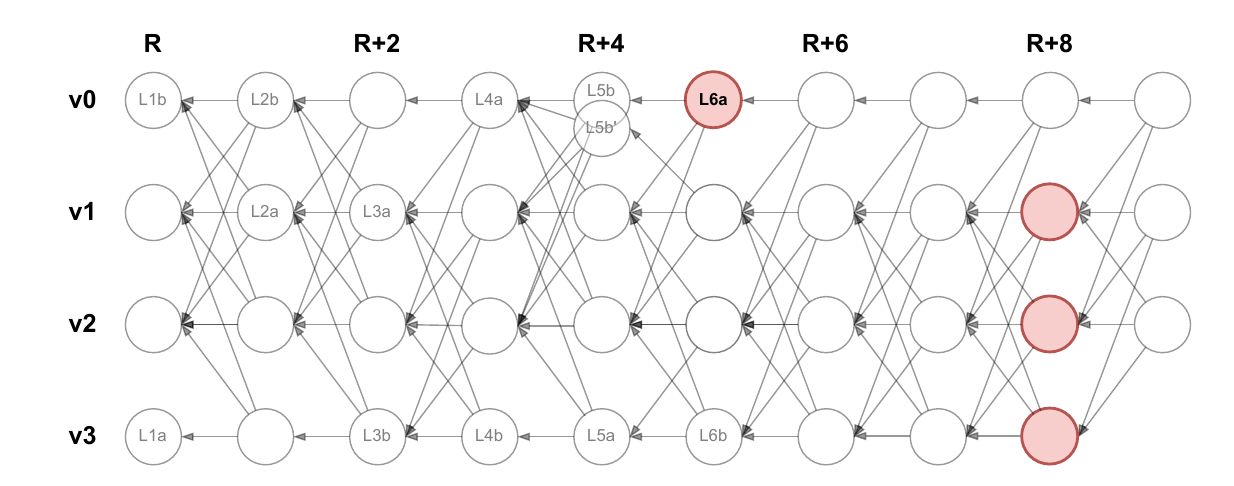}
        \caption{Example of direct decision rule. $L_{6a}$ is classified as \sskip.}
        \label{fig:example-2}
    \end{subfigure}
    \begin{subfigure}[t]{\textwidth}
        \centering
        \includegraphics[width=0.75\textwidth]{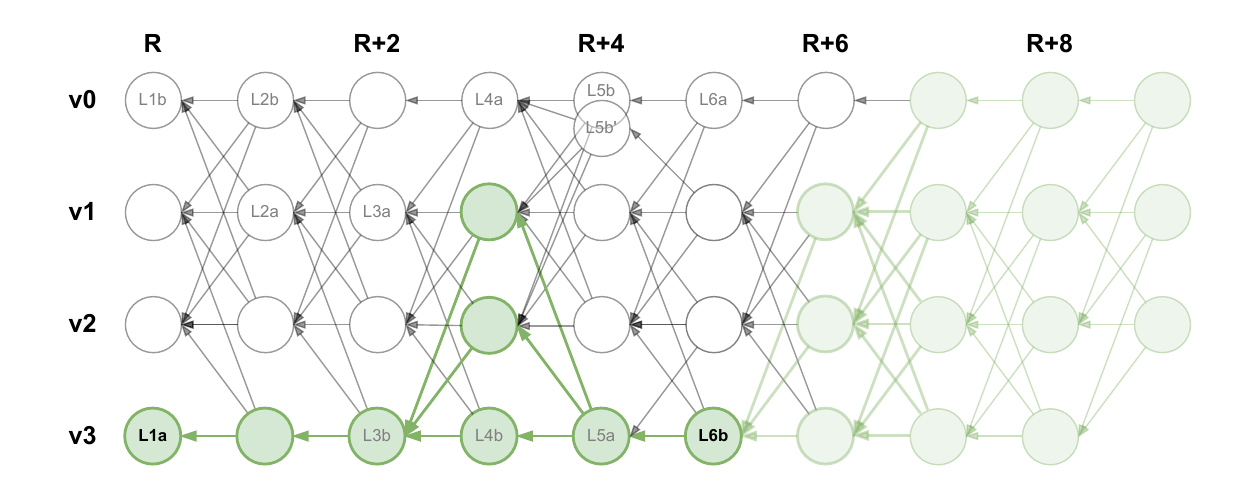}
        \caption{Example of indirect decision rule. $L_{1b}$ is classified as \scommit.}
        \label{fig:example-3}
    \end{subfigure}
    \caption{
        Example of \sysname execution with 4 validators, configured with a wave length of 5 rounds and 2 leader slots per round.
    }
    \label{fig:example-detailed}
\end{figure}

The first step for the validator is to identify the leader slots by reconstructing the global perfect coin for each round. As described in \Cref{sec:decision-rule} (step 1), the validator derives the following leader slots:
$$ [L_{6b}, L_{6a}, L_{5b}, L_{5a}, L_{4b}, L_{4a}, L_{3b}, L_{3a}, L_{2b}, L_{2a}, L_{1b}, L_{1a}] $$

Next, the validator attempts to classify each leader slot as either \scommit or \sskip using the \emph{direct decision rule} (step 2), starting with the highest slot, $L_{6b}$. As shown in \Cref{fig:example-1}, the validator classifies $L_{6b}$ as \scommit since it is certified by $2f+1$ blocks from round $R+9$, specifically $B_{(v_0,R+9)}$, $B_{(v_1,R+9)}$, and $B_{(v_2,R+9)}$.

The validator then proceeds to $L_{6a}$. As illustrated in \Cref{fig:example-2}, it classifies this block as \sskip because $2f+1$ blocks from round $R+8$ ($B_{(v_1, R+8)}$, $B_{(v_2, R+8)}$, and $B_{(v_3, R+8)}$) do not vote for $L_{6a}$.

Next, the validator examines $L_{5b}$. In its local view of the DAG, it encounters two equivocations for this leader slot: $L_{5b}$ and $L_{5b}'$. The validator then invokes the function $\Call{IsVote}{\cdot}$ (shown in Line~\ref{alg:line:isvote} of \Cref{alg:helper}, \Cref{sec:algorithms}) to determine which of these equivocations, if any, receives votes from the blocks of round $R+7$. For each block of this round, the validator conducts a depth-first search starting from the block, following its hash references to see if it first encounters $L_{5b}$ or $L_{5b}'$.

In this example, the validator first targets $B_{(v_0,R+7)}$ and finds it votes for $L_{5b}$. Upon targeting $B_{(v_1,R+7)}$, it discovers a vote for $L_{5b}'$. Continuing this process with $B_{(v_2,R+7)}$ and $B_{(v_3,R+7)}$ also leads to a vote for $L_{5b}'$. Consequently, since there are $2f+1$ blocks from round $R+7$ that do not vote for $L_{5b}$, the validator classifies it as \sskip. And since there are at least $2f+1$ blocks from round $R+8$ ($B_{(v_0,R+8)}$, $B_{(v_1,R+8)}$, $B_{(v_2,R+8)}$, and $B_{(v_3,R+8)}$) that certify $L_{5b}'$, the validator classifies it as \scommit.

The validator then moves on to $L_{5a}$, classifying it as \scommit as it receives sufficient certification from blocks of round $R+8$ ($B_{(v_0,R+8)}$, $B_{(v_1,R+8)}$, $B_{(v_2,R+8)}$, and $B_{(v_3,R+8)}$).
Similarly, the validator classifies both $L_{4b}$ and $L_{4a}$ as \scommit since they are also certified by $B_{(v_0,R+7)}$, $B_{(v_1,R+7)}$, $B_{(v_2,R+7)}$, and $B_{(v_3,R+7)}$.
This reasoning applies to the slots $L_{3b}$ and $L_{3a}$, which are certified by blocks of round $R+6$, as well as to $L_{2b}$ and $L_{2a}$, certified by blocks of round $R+5$.
Finally, $L_{1b}$ is certified by $B_{(v_0,R+4)}$ (through both $L_{5b}$ and $L_{5b}'$), $B_{(v_1,R+4)}$, $B_{(v_2,R+4)}$, and $B_{(v_3,R+4)}$.

However, the direct commit rule fails to classify $L_{1a}$. There are neither $2f+1$ blocks from round $R+3$ that do not vote for it (which would classify it as \sskip) nor $2f+1$ blocks from round $R+4$ that certify it (preventing its classification as \scommit). Therefore, the validator turns to the \emph{indirect decision rule} (step 3) to classify $L_{1a}$.

As specified in \Cref{sec:decision-rule}, the validator first seeks the anchor for $L_{1a}$, which is $L_{6b}$. The anchor is defined as the first block with a round number $r' > r+5$ that is classified as either \sundecided or \scommit. Consequently, $L_{6a}$ cannot serve as the anchor for $L_{1b}$, making $L_{6b}$ its anchor.
The validator then checks for an existing certificate over $L_{1a}$ that is referenced by the causal history of its anchor, $L_{6b}$. As illustrated in \Cref{fig:example-3}, in this case, $L_{1b}$ is certified by $L_{5a}$, which references $2f+1$ votes for $L_{1b}$ ($B_{(v_1,R+3)}$, $B_{(v_2,R+3)}$, and $B_{(v_3,R+3)}$). Thus, $L_{1b}$ is classified as \scommit. Without such a certificate, the validator would have classified $L_{1b}$ as \sskip.

%% file: sections/proofs.tex
\section{Security Proofs}\label{sec:proofs}

This section proves the correctness of \sysname, by showing that \sysname satisfies the properties of Byzantine Atomic Broadcast (BAB) from \Cref{sec:overview}. We prove the correctness of both the $4$-round and $5$-round versions of \sysname. We start with results that hold for both versions (these are mostly safety-related results) in \Cref{app:common-proofs} and continue with version-specific results in Appendices~\ref{app:w5-proofs} and~\ref{app:w4-proofs}.

\subsection{Common Proofs for $w = 4$ and $w = 5$}\label{app:common-proofs}

We start by proving the Total Order and Integrity properties of BAB. A crucial intermediate result towards these properties is that all honest validators have consistent commit sequences, i.e., the committed sequence of one honest validator is a prefix of another's, or vice-versa. This is shown in Lemmas~\ref{lem:agree-commit} and~\ref{lem:consistent}, which the following lemmas and observations build up to.

\begin{lemma} \label{lma1_cert_path}
    If in round $r$, $2f+1$ blocks from distinct validators certify a block $b$, then all blocks at future rounds $r'>r$ will have a path to a certificate for $b$ from round $r$.
\end{lemma}
\begin{proof}
    We prove the lemma by induction on $r'$. The base case is $r' = r+1$. Let $b'$ be a block at round $r'$. Since $b'$ points to $2f+1$ blocks at round $r$, by quorum intersection, $b'$ must point to at least one of the certificates for $b$.

    For the induction case, assume the lemma holds up to round $r'$ and consider the case of round $r'+ 1$. Let $b'$ be a block at round $r'+1$. By the induction hypothesis, $2f+1$ blocks at round $r'$ have paths to round-$r$ certificates for $b$. Since $b'$ points to $2f+1$ blocks from round $r'$, by quorum intersection, $b'$ must point to at least one block that has a path to a round-$r$ certificate for $b$.
\end{proof}

\begin{observation}\label{obs:no-equiv-vote}
    A block cannot vote for more than one block proposal from a given validator, in a given round.
\end{observation}
\begin{proof}
    This is by construction. Honest validators interpret support in the DAG through deterministic depth-first traversal. So even if a block $b$ in the vote round has paths to multiple leader round blocks from the same validator $v$ (i.e., equivocating blocks), all honest validators will interpret $b$ to vote for only one of $v$'s blocks (the first block to appear in the depth-first traversal starting from $b$).
\end{proof}

\begin{lemma}\label{lem:one-cert-per-validator}
    At most a single block per round from the same validator can be certified.
\end{lemma}
\begin{proof}
    Assume by contradiction that in a given round $r$, there exist two distinct blocks $b$ and $b'$ from the same validator $v$ such that both $b$ and $b'$ are certified. This means that there exist round-$(r+w-1)$ blocks $c_b$ and $c_{b'}$ that certify $b$ and $b'$, respectively. $c_b$ and $c_{b'}$ must point to $2f+1$ votes for $b$ and $b'$, respectively. By quorum intersection, there exists an honest validator that has voted for both $b$ and $b'$ in the vote round. Since honest validators only produce a single block per round, this implies that there exists a block that votes for both $b$ and $b'$, contradicting \Cref{obs:no-equiv-vote}.

\end{proof}

\begin{observation}\label{obs:exists-certificate}
    If an honest validator $v$ directly or indirectly commits a block $b$, then $v$'s local DAG contains a certificate for $b$.
\end{observation}
\begin{proof}
    This follows immediately from our direct and indirect commit rules.
\end{proof}

\begin{observation}\label{obs:agree-order}
    Honest validators agree on the sequence of leader slots.
\end{observation}
\begin{proof}
    This follows immediately from the properties of the common coin, see~\Cref{sec:model}.
\end{proof}

\begin{lemma} \label{lma4_honest_unique}
    If an honest validator $v$ commits some block $b$ in a slot $s$, then no other honest validator decides to directly skip the slot $s$.
\end{lemma}
\begin{proof}
    Assume by contradiction that some honest validator $v'$ decides to directly skip $s$. Then it must be the case that in the local DAG of $v'$, at least $2f+1$ validators did not vote for $b$. However, since $v$ commits $b$ at $s$, by \Cref{obs:exists-certificate}, there must exist a certificate for $b$ at $s$. So in $v$'s local DAG there must be $2f+1$ validators that vote for $b$. By quorum intersection, at least one honest validator both voted for $b$ and did not vote for $b$. Since honest validators produce a single block in the vote round, this is a contradiction.
\end{proof}

\begin{lemma} \label{lma5_honest_skip}
    If an honest validator directly commits some block in a slot $s$, then no other honest validator decides to skip the slot $s$.
\end{lemma}
\begin{proof}
    Assume by contradiction that an honest validator $v$ directly commits block $b$ in slot $s$ while another honest validator $v'$ decides to skip $s$. By \Cref{lma4_honest_unique}, $v'$ cannot directly skip $s$; it must be the case therefore that $v'$ skips $s$ using the indirect decision rule. Let $r$ be the round of $s$. Since $v$ directly commits $b$, there exist $2f+1$ certificates for $b$ at $s$. Therefore, by \Cref{lma1_cert_path}, all blocks at rounds $r' > r+w-1$, including the anchor of $s$, have a path to a certificate for $b$ at $s$. Thus, $v'$ cannot decide to skip $s$ using the indirect decision rule. We have reached a contradiction.
\end{proof}

\begin{lemma}\label{lem:agree-commit}
    If a slot $s$ is committed at two honest validators, then $s$ contains the same block at both validators.
\end{lemma}
\begin{proof}
    Let $v$ and $u$ be two honest validators and assume that $v$ commits block $b$ at slot $s$. We will show that if $u$ commits slot $s$, then $s$ contains $b$ at $s$. Let $w$ be the validator that produced block $b$. By \Cref{obs:exists-certificate}, for $b$ to be committed at slot $s$ at $v$, there must exist at least one certificate for $b$. By \Cref{obs:agree-order}, $v$ and $u$ agree that $s$ must contain a block by $w$. By \Cref{lem:one-cert-per-validator}, at most a single block per round from $w$ can be certified. So $u$ cannot have a certificate for any other block than $b$ at slot $s$.
\end{proof}

We say that a slot is \textit{decided} at a validator $v$ if $s$ is committed or skipped, \ie if it is categorized as $\scommit$ or $\sskip$. Otherwise, $s$ is \textit{undecided}.
\begin{lemma}\label{lem:consistent}
    If a slot $s$ is decided at two honest validators $v$ and $v'$, then either both validators commit $s$, or both validators skip $s$.
\end{lemma}
\begin{proof}
    Assume by contradiction that there exists a slot $s$ such that $v$ and $v'$ decide differently at $s$. We consider a finite execution prefix and assume \textit{wlog} that $s$ is the highest slot at which $v$ and $v'$ decide differently (*). Further assume \textit{wlog} that $v$ commits $s$ and $v'$ skips $s$. By \Cref{lma4_honest_unique} and \Cref{lma5_honest_skip}, neither $v$ nor $v'$ could have used the direct decision rule for $s$; they must both have used the indirect rule. Consider now the anchor of $s$: $v$ and $v'$ must agree on which slot is the anchor of $s$, since by our assumption (*) above, they make the same decisions for all slots higher than $s$, including the anchor of $s$. Let $s'$ be the anchor of $s$; $s'$ must be committed at both $v$ and $v'$. Thus, by \Cref{lem:agree-commit}, $v$ and $v'$ commit the same block $b'$ at $s'$. But then $v$ and $v'$ cannot reach different decisions about slot $s$ using the indirect decision rule. We have reached a contradiction.
\end{proof}


We have proven the consistency of honest validators' commit sequences: honest validators commit (or skip) the same leader blocks, in the same order. However, we are not done: we also need to prove that non-leader blocks are delivered in the same order by honest validators. We show this next.

\para{Causal history \& delivery conditions}
Consider an honest validator $v$. We call the \textit{causal history} of a block $b$ in $v$'s DAG, the transitive closure of all blocks referenced by $b$ in $v$'s DAG, including $b$ itself. In \sysname, a block $b$ is delivered by an honest validator $v$ if (1) there exists a committed leader block $l$ in $v$'s DAG such that $b$ is in $l$'s causal history (2) all slots up to $l$ are decided in $v$'s DAG and (3) $b$ has not been delivered as part of a lower slot's causal history. In this case we say $b$ is \textit{delivered at} slot $s$, or \textit{delivered with} block $l$.



\begin{lemma}\label{lem:delivered-same-slot}
    If a block $b$ is delivered by two honest validators $v$ and $v'$, then $b$ is delivered at the same slot $s$, and $b$ is delivered with the same leader block $l$, at both $v$ and $v'$.
\end{lemma}
\begin{proof}
    Let $s$ be the slot at which $b$ is delivered at validator $v$, and $l$ the corresponding leader block in $s$, also at validator $v$. Consider now the slot $s'$ at which $b$ is delivered at validator $v'$, and $l'$ the corresponding leader block. Assume by contradiction that $s' \ne s$. If $s' < s$, then $v$ would have also delivered $b$ at slot $s'$, since by \Cref{lem:agree-commit} must commit the same leader blocks in the same slots, so $v$ could not have delivered $b$ again at slot $s$; a contradiction. Similarly, if $s < s'$, then $v'$ would have already delivered $b$ at slot $s$, since by \Cref{lem:agree-commit} $v$ and $v'$ must have committed the same block in slot $s$; contradiction. Thus it must be that $s = s'$, and by \Cref{lem:agree-commit}, $l = l'$.
\end{proof}

We can now prove the main safety properties of \sysname: Total Order and Integrity.
\begin{theorem}[Total Order]
    \sysname satisfies the total order property of Byzantine Atomic Broadcast.
\end{theorem}
\begin{proof}
    This property follows immediately from \Cref{lem:delivered-same-slot} and the fact that honest validators order the causal histories of committed blocks using the same deterministic function, and deliver blocks in this order.
\end{proof}

\begin{theorem}[Integrity]\label{thm:integrity}
    \sysname satisfies the integrity property of Byzantine Atomic Broadcast.
\end{theorem}
\begin{proof}
    This is by construction: a block $b$ is delivered as part of the causal history of a committed leader block only if $b$ has not been delivered along with an earlier leader block (see "Causal history \& delivery conditions" above). So an honest validator cannot deliver the same block twice.
\end{proof}

We now turn to liveness properties. The following two lemmas establish that blocks broadcast by honest validators are eventually included in all honest validators' DAGs.

\begin{lemma}\label{lem:eventually-all-references}
    If a block $b$ produced by an honest validator $v$ references some block $b'$, then $b'$ will eventually be included in the local DAG of every honest validator.
\end{lemma}
\begin{proof}
    This is ensured by the synchronizer sub-component in each validator: if some validator $w$ receives $b$ from $v$, but does not have $b'$ yet, $w$ will request $b'$ from $v$; since $v$ is honest and the network links are reliable, $v$ will eventually receive $w$'s request, send $b'$ to $w$, and $w$ will eventually receive $b'$. The same is recursively true for any blocks from the causal history of $b'$, so $w$ will eventually receive all blocks from the causal history of $b'$ and thus include $b'$ in its local DAG.
\end{proof}

\begin{lemma}\label{lem:eventually-include}
    If an honest validator $v$ broadcasts a block $b$, then every correct validator will eventually include $b$ in its local DAG.
\end{lemma}
\begin{proof}
    Since network links are reliable, all honest validators will eventually receive $b$ from $v$. By \Cref{lem:eventually-all-references}, all honest validators will eventually receive all of $b$'s causal history, and so will include $b$ in their local DAG.
\end{proof}

The following crucial lemma establishes that in any round $r$, there is at least one block $b$, called a common core, such that all blocks at round $r+2$ have a path to $b$.

\begin{lemma}\label{lem:common-core-1}
    For any $r$, there is at least one block $b$ from round $r$ such that any valid block from round $r+2$ has a path to $b$.
\end{lemma}
\begin{proof}
    Consider a set $B$ of $2f+1$ blocks in round $r+1$ from  honest validators.
    Using $B$, we create a table $T$, as follows: for blocks $b,c\in B$, let $T[b,c]=1$ if $b$ in $r+1$ references $c$ in $r$, $T[b,c]=0$ otherwise. By quorum intersection, any $b$ will reference at least $f+1$ blocks in round $r$ that are also in $B$, so each row of $T$ has at least $f+1$ entries equal to $1$. Thus, $T$ has at least $(2f+1)(f+1)$ entries equal to $1$. By a counting argument, there is a block $c^*$ in $B$ that has a $1$ entry in at least $f+1$ rows, i.e., a block from round $r$ which is referenced by $f+1$ blocks from round $r+1$. Let $P'$ be the set of blocks from round $r+1$ which reference $c^*$. Consider now any valid block $b$ in round $r+2$; $b$ references $2f+1$ blocks in $r+1$, so by quorum intersection $b$ references at least one block in $B$. Thus, $b$ has a path to $c^*$.
\end{proof}

\subsection{Specific Proofs for $w = 5$}\label{app:w5-proofs}

We continue with proofs that are specific to the liveness of the $w=5$ version of \sysname. We show that each wave has at least $2f+1$ leader blocks that can be directly committed (Lemmas~\ref{lem:common-core-2} and \ref{lem:common-core-3}), and thus that each wave has a nonzero probability of directly committing at least one block (\Cref{lem:commit-prob}). We then show that each slot is eventually decided directly or indirectly (\Cref{lem:all-slots-decide}). Finally, we show that \sysname satisfies the Validity and Agreement properties of BAB.

As a consequence of \Cref{lem:common-core-1}, we have the following:

\begin{lemma}\label{lem:common-core-2}
    For any $r$, there exists a set $S$ of at least $2f+1$ blocks from round $r$ such that any valid block from round $r+3$ is a vote for every block in $S$.
\end{lemma}
\begin{proof}
    Let $r' = r+1$. By \Cref{lem:common-core-1}, there exists a block $b$ in round $r' = r+1$ such that any valid block from round $r' + 2 = r + 3$ has a path to $b$. Now let $S$ be the set of blocks referenced by the block $b$. $S$ must contain at least $2f+1$ blocks from round $r$. Every block from round $r+3$ has a path to $b$ and thus, through $b$, to every block in $S$.
\end{proof}

From this we can derive the following crucial lemma:
\begin{lemma}\label{lem:common-core-3}
    For any $r$, there exists a set $S$ of at least $2f+1$ blocks from round $r$ such that every block in $S$ has at least $2f+1$ certificates in round $r+4$.
\end{lemma}
\begin{proof}
    Take $S$ to be the set from \Cref{lem:common-core-2}. There are at least $2f+1$ blocks in $r+4$. Any block $b$ in round $r+4$ must reference $2f+1$ blocks from round $r+3$. By \Cref{lem:common-core-2}, every block from round $r+3$ is a vote for every block in $S$, so $b$ must be a certificate for every block in $S$.
\end{proof}

We denote by $\ell \leq 3f+1$ the number of leader slots per round.
\begin{lemma}\label{lem:commit-prob}
    Fix a round $r$. If $\ell >f$, then an honest validator directly commits at least one slot corresponding to round $r$. Otherwise, the probability that an honest validator directly commits at least one slot corresponding to round $r$ is at least $p^\star = 1-\frac{\binom{f}{\ell}}{\binom{3f+1}{\ell}} > 0$.
\end{lemma}
\begin{proof}
    By \Cref{lem:common-core-3}, at least $2f+1$ blocks from round $r$ can be directly committed, out of a maximum of $3f+1$ blocks. When the common coin is released in round $r+4$, it selects uniformly at random $\ell$ round-$r$ blocks as the $\ell$ slots of round $r$.

    In the case $\ell >f$, by quorum intersection, there exists at least one slot selected by the common coin among the $2f+1$ blocks that can be directly committed.

    In the case $\ell \leq f$, we can model the number of directly committed slots in round $r$ as a hypergeometric random variable, where a success event corresponds to selecting a slot that can be directly committed. The probability of $0$ successes (i.e., not committing any slots directly) is therefore at most $\frac{\binom{f}{\ell}}{\binom{3f+1}{\ell}} < 1$.
\end{proof}

\begin{lemma}\label{lem:all-slots-decide}
    Fix a slot $s$. Every honest validator eventually either commits or skips $s$, with probability $1$.
\end{lemma}
\begin{proof}
    We prove the lemma by showing that the probability of $s$ remaining undecided forever at some honest validator is $0$. In order for $s$ to remain undecided forever, $s$ cannot be committed or skipped directly. Furthermore, $s$ cannot be decided using the indirect rule. This means that the anchor $s'$ of $s$ must also remain undecided forever, and therefore the anchor $s''$ of $s'$ must remain undecided forever, and so on. The probability of this occurring is at most equal to the probability of an infinite sequence of rounds with no directly committed slots, equal to $\lim_{t\rightarrow\infty}(1 - p^\star)^t = 0$, where $p^\star > 0$ is the probability from \Cref{lem:commit-prob}.
\end{proof}

\begin{theorem}[Validity]
    \sysname satisfies the validity property of Byzantine Atomic Broadcast.
\end{theorem}
\begin{proof}
    Let $v$ be an honest validator and $b$ a block broadcast by $v$. We show that, with probability $1$, $b$ is eventually delivered by every honest validator. By \Cref{lem:eventually-include}, $b$ is eventually included in the local DAG of every honest validator. So every honest validator will eventually include a reference to $b$ in at least one of its blocks. Let $r$ be the highest round at which some honest validator includes a reference to $b$ in one of its blocks. By \Cref{lem:commit-prob}, with probability $1$, eventually some block $b'$ at a round $r' > r$ will be directly committed. Block $b'$ must reference at least $2f+1$ blocks, thus at least $f+1$ blocks from honest validators. Since all validators have $b$ in their causal histories by round $r$, $b'$ must therefore have a path to $b$. \Cref{lem:all-slots-decide} guarantees that all slots before $b'$ are eventually decided, so $b'$ is eventually delivered. Thus, $b$ will be delivered at all honest validators at the latest when $b'$ is delivered along with its causal history.
\end{proof}

\begin{theorem}[Agreement]
    \sysname satisfies the agreement property of Byzantine Atomic Broadcast.
\end{theorem}
\begin{proof}
    Let $v$ be an honest validator and $b$ a block delivered by $v$. We show that, with probability $1$, $b$ is eventually delivered by every honest validator. Let $l$ be the leader block with which $b$ is delivered, and $s$ the corresponding slot. By \Cref{lem:all-slots-decide}, all blocks up to and including $s$ are eventually decided by all honest validators, with probability $1$. By \Cref{lem:agree-commit}, all honest validators commit $l$ in $s$. Therefore, all honest validators deliver $b$ eventually.
\end{proof}

\subsection{Specific Proofs for $w = 4$}\label{app:w4-proofs}

We now turn to the liveness of the $w = 4$ version of \sysname. We first show that under an asynchronous network, liveness is guaranteed, albeit with a smaller probability of direct commit at each wave than in the $w=5$ version. We later show that under a random network, \sysname directly commits all valid leader blocks in each wave with high probability. Recall that in the random network model, a valid block from round $r+1$ references a set of $2f+1$ valid blocks from round $r$, sampled uniformly at random among all valid round-$r$ blocks.

\begin{lemma}\label{lem:common-core-w4-1}
    For any $r$, there exists a block $b$ from round $r$ such that $b$ has at least $2f+1$ certificates in round $r+3$.
\end{lemma}
\begin{proof}
    By \Cref{lem:common-core-1}, there exists $b$ in round $r$ such that any block in round $r+2$ has a path to $b$, and therefore is a vote for $b$. Since any block in round $r+3$ must reference $2f+1$ blocks in round $r+2$, any block in round $r+3$ must be a certificate for $b$. Since every honest validator publishes a block in round $r+3$, there must exist at least $2f+1$ certificates for $b$ in round $r+3$.
\end{proof}

We again denote by $\ell \leq 3f+1$ as the number of leader slots per round.
\begin{lemma}\label{lem:common-core-w4-2}
    Assume the asynchronous network model and fix a round $r$. If $\ell = 3f+1$, then an honest validator commits at least one slot corresponding to round $r$. Otherwise, the probability that an honest validator directly commits at least one slot corresponding to round $r$ is at least $p^\star = \frac{\ell}{3f+1}$.
\end{lemma}
\begin{proof}
    By \Cref{lem:common-core-w4-1}, there exists at least one block $b$ in round $r$ that can be directly committed. When the common coin is released in round $r+3$, it selects uniformly at random $\ell$ round-$r$ blocks as the $\ell$ slots of round $r$.

    If $\ell=3f+1$, then all possible blocks of round $r$ are included in the slots, and thus $b$ is directly committed.

    In the case $\ell < 3f+1$, we can model the number of directly committed slots in round $r$ as a hypergeometric random variable, where a success event corresponds to selecting the slot that can be directly committed. There are $3f+1$ states in total, out of which only $1$ is a success state; and there are $\ell$ draws. The probability of one success is therefore $\frac{\binom{1}{1}\binom{3f}{\ell - 1}}{\binom{3f+1}{\ell}} = \frac{\ell}{3f+1}$.
\end{proof}

\begin{lemma}\label{lem:random-common-core}
    In the random network model, with high probability, every block in round $r+2$ is a vote for every block in round $r$.
\end{lemma}
\begin{proof}
    We prove the lemma by showing, through Markov's inequality, that the probability of any block in round $r$ being unreachable from any block in round $r+2$ approaches 0 exponentially in $f$.

    Take a pair of blocks $b_r$ and $b_{r+2}$ in rounds $r$ and $r+2$, respectively. We compute the probability that there is no path from $b_{r+2}$ to $b_r$. Block $b_{r+2}$ must reference $2f+1$ blocks from round $r+1$; let $b_{r+1}$ be one such block. The probability that $b_{r+1}$ references $b_r$ is at least $p = \frac{2f+1}{3f+1}$, since $b_{r+1}$ references $2f+1$ randomly selected blocks from round $r$. The probability that there is no path from $b_{r+2}$ to $b_r$ is therefore at most $q = (1 - p)^{2f+1}$.

    We now compute the expected number of pairs of blocks $b_r$ and $b_{r+2}$ from rounds $r$ and $r+2$, respectively, such that $b_r$ is not reachable from $b_{r+2}$. There are at most $(3f+1)^2$ pairs, and each pair is not connected by a path with probability at most $q$, thus the expected number of unreachable pairs is at most $E = q(3f+1)^2 = (3f+1)^2 (1 - p)^{2f+1}$.

    Using Markov's inequality, the probability that there exists at least one unreachable pair is: $$\Pr(\text{At least one unreachable pair}) \leq E = (3f+1)^2 (1 - p)^{2f+1}.$$ As $f$ increases, this probability rapidly approaches $0$ due to the exponential term.
\end{proof}

\begin{lemma}\label{lem:random-common-core-2}
    Assume the random network model and fix a round $r$. With high probability, an honest validator directly commits every leader slot chosen by the common coin in round $r+3$.
\end{lemma}
\begin{proof}
    By \Cref{lem:random-common-core}, every block in round $r+2$ is a vote for every block in round $r$ with high probability. Thus every block in round $r+3$ is a certificate for every block in round $r$ with high probability. So every honest validator sees $2f+1$ certificates for every block in round $r$ and thus sees also for the blocks chosen by the common coin at least $2f+1$ certificates.
\end{proof}

\begin{lemma}\label{lem:all-slots-decide-2}
    Fix a slot $s$. In both the asynchronous network model and the random network model, every honest validator eventually either commits or skips $s$ with probability~$1$.
\end{lemma}
\begin{proof}
    The proof is analogous to the proof of \Cref{lem:all-slots-decide}. By \Cref{lem:common-core-w4-2} and \Cref{lem:random-common-core-2}, the probability of a honest validator directly committing any leader block in a given round is greater than $0$, in both the asynchronous and the random network models. Thus, the probability of an infinite sequence of rounds without any directly committed blocks is $0$. This implies that every slot that is not directly decided will eventually have a committed anchor and become decided itself.
\end{proof}

\begin{theorem}[Validity]
    \sysname satisfies the validity property of Byzantine Atomic Broadcast.
\end{theorem}
\begin{proof}
    The proof is similar to the proof of validity in the $w=5$ case. Let $v$ be an honest validator and $b$ a block broadcast by $v$. We show that, with probability $1$, $b$ is eventually delivered by every honest validator. By \Cref{lem:eventually-include}, $b$ is eventually included in the local DAG of every honest validator. So every honest validator will eventually include a reference to $b$ in at least one of its blocks. Let $r$ be the highest round at which some honest validator includes a reference to $b$ in one of its blocks. By \Cref{lem:common-core-w4-2} and \Cref{lem:random-common-core-2}, with probability $1$, eventually some block $b'$ at a round $r' > r$ will be directly committed. Block $b'$ must reference at least $2f+1$ blocks, thus at least $f+1$ blocks from honest validators. Since all validators have $b$ in their causal histories by round $r$, $b'$ must therefore have a path to $b$. \Cref{lem:all-slots-decide-2} guarantees that all slots before $b'$ are eventually decided, so $b'$ is eventually delivered. Thus, $b$ will be delivered at all honest validators at the latest when $b'$ is delivered along with its causal history.
\end{proof}

\begin{theorem}[Agreement]
    \sysname satisfies the agreement property of Byzantine Atomic Broadcast.
\end{theorem}
\begin{proof}
    The proof is similar to the proof of agreement in the $w=5$ case. Let $v$ be an honest validator and $b$ a block delivered by $v$. We show that, with probability $1$, $b$ is eventually delivered by every honest validator. Let $l$ be the leader block with which $b$ is delivered and $s$ the corresponding slot. By \Cref{lem:all-slots-decide-2}, all blocks up to and including $s$ are eventually decided by all honest validators, with probability $1$. By \Cref{lem:agree-commit}, all honest validators commit $l$ in $s$. Therefore, all honest validators deliver $b$ eventually.
\end{proof}

\para{Note: \sysname with $w=3$}
It is possible to configure \sysname with $3$-round waves, by removing all \boost rounds, and keeping the \propose round ($r$), \vote round ($r+1$) and \certify round ($r+2$). Such a protocol would still satisfy safety, as all results up to and including \Cref{thm:integrity} hold when $w=3$. However, this $3$-round version of \sysname would no longer satisfy liveness, because the common core approach in \Cref{lem:common-core-1} can no longer be used to guarantee that at least one leader block can be directly committed in each wave.

%% file: sections/appendix-eval.tex
\section{\sysname Impact of Multi-leader} \label{sec:eval-impact-multi-leader}
This section completes \Cref{sec:evaluation} by presenting the impact of the number of leaders per round in \sysname when implemented with $5$ rounds per wave.
\Cref{fig:leaders-w-5} illustrates how \sysname configured with a wave length of 5 rounds performs with 1, 2, and 3 leaders per round under both normal conditions and scenarios involving 3 crash faults.
We observe a latency reduction as the number of leaders increases similar to \Cref{fig:leaders} (\Cref{sec:evaluation}).
Specifically, when the number of leaders rises from 1 to 3, \sysname's average latency decreases by approximately 40ms in ideal scenario,
and by approximately 100ms in the crash failure scenario.
Similarly to \Cref{fig:leaders}, increasing the number of leaders beyond 3 did not further decrease latency.

\begin{figure}[t]
    \vskip -1em
    \centering
    \includegraphics[width=0.7\linewidth]{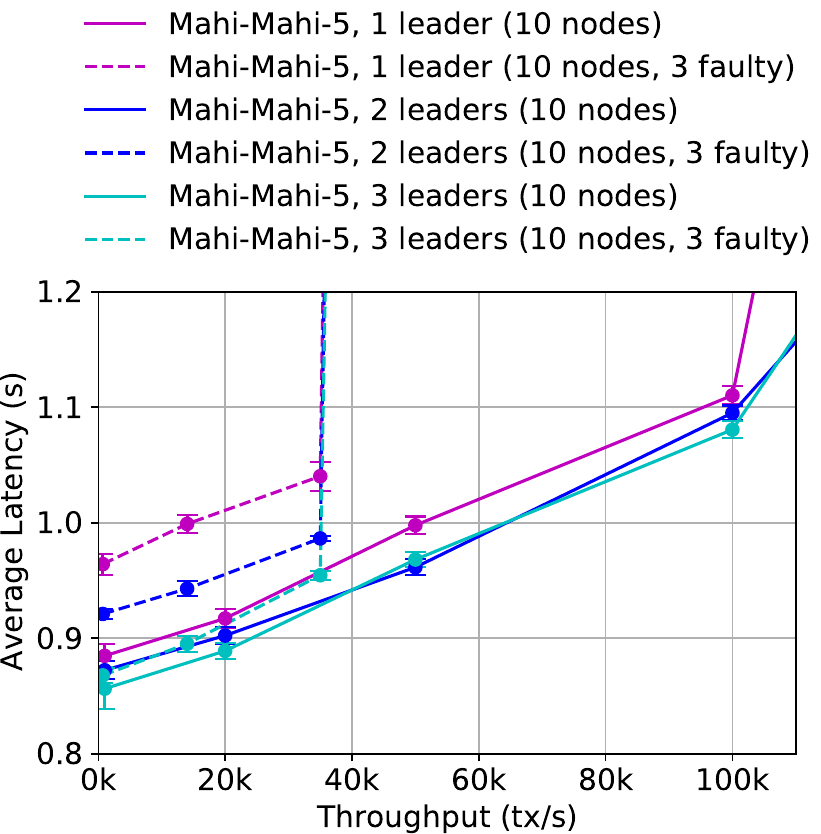}
    \caption{
        Impact of the number of leaders per round in \sysname. WAN measurements with 10 validators. Zero and three faults. 512B transaction size.
    }
    \label{fig:leaders-w-5}
    \vskip -1em
\end{figure}